\documentclass[12pt]{article}
\usepackage{amsmath}%
\usepackage{amsfonts}%
\usepackage{amssymb}%
\usepackage{amsthm}
\usepackage{graphicx}
\usepackage{tikz-cd}
\usepackage{hyperref}
\usepackage[T1]{fontenc}
\usepackage{xcolor}[table]
\usepackage{array}
\usepackage{mathtools}

\usepackage{geometry}
\theoremstyle{plain}

\newtheorem*{acknowledgement*}{Acknowledgement}

\newtheorem{definition}{Definition}

\numberwithin{equation}{section}
\newtheorem{theorem}{Theorem}[section]
\newtheorem{lemma}[theorem]{Lemma}
\newtheorem{proposition}[theorem]{Proposition}
\newtheorem{corollary}[theorem]{Corollary}

\theoremstyle{remark}
\newtheorem{remark}[theorem]{Remark}
\newtheorem{example}[theorem]{Example}
\numberwithin{equation}{section}

\newcommand{\Dsxy}{\Delta^{\mathfrak{S} \times \mathfrak{X} \times \mathfrak{Y}}}

\newcommand{\E}{\mathbb{E}}
\newcommand{\Prob}{\mathbb{P}}

\title{How do correlations shape the landscape of information?}
\author{Huang Ching-Peng}

\date{ }

\begin{document}

\maketitle

\begin{abstract}
    We explore a few common models on how correlations affect information. The main model considered is the Shannon mutual information $I(S:R_1,\cdots, R_i)$ over distributions with marginals $P_{S,R_i}$ fixed for each $i$, with the analogy in which $S$ is the stimulus and $R_i$'s are neurons. We work out basic models in details, using algebro-geometric tools to write down discriminants that separate distributions with distinct qualitative behaviours in the probability simplex into toric chambers and evaluate the volumes of them algebraically.  Some algebro-geometric structure suitable for the framework is explored. 
    
    We hope this paper serves for communication between communities especially mathematics and theoretical neuroscience on the topic.
    \vspace{0.8 cm}

    
\end{abstract}

\noindent\textbf{Acknowledgements} The authour would like to thank Prof Stefano Panzeri for suggesting initial questions and discussions on the relationship between the two information decompositions that lead to this project, as well as generous support for the position, Marco Celotto for several sessions of in-depth meetings, and Simone Blanco Malerba for suggesting related articles.

\tableofcontents

\section{Introduction: battles between synergy and redundancy} 

Given three random variables $S, X, Y$, independence of the pairs 
\begin{equation} \label{marg indep}
    X\perp S \text{ and } Y \perp S,
\end{equation} as we know, does not imply the joint $(X,Y)$ is independent of $S$. If we consider Shannon's mutual information, how does the ``interaction'' between $X$ and $Y$ exert different amount of $I(S:X,Y)$?  In Fig. \ref{diag 0}, we see the landscape of this mutual information on the domain of all possible distributions satisfying the condition \ref{marg indep}. Conceptually, since $I(S:X)$ and $I(S:Y)$ are zero, we want to think of $I(S:X,Y)$ the information purely coming from the \textit{second order interaction of $X,Y$ about $S$}. 

\begin{figure}\label{diag 0}
    \centering
    \includegraphics[width=0.5\linewidth]{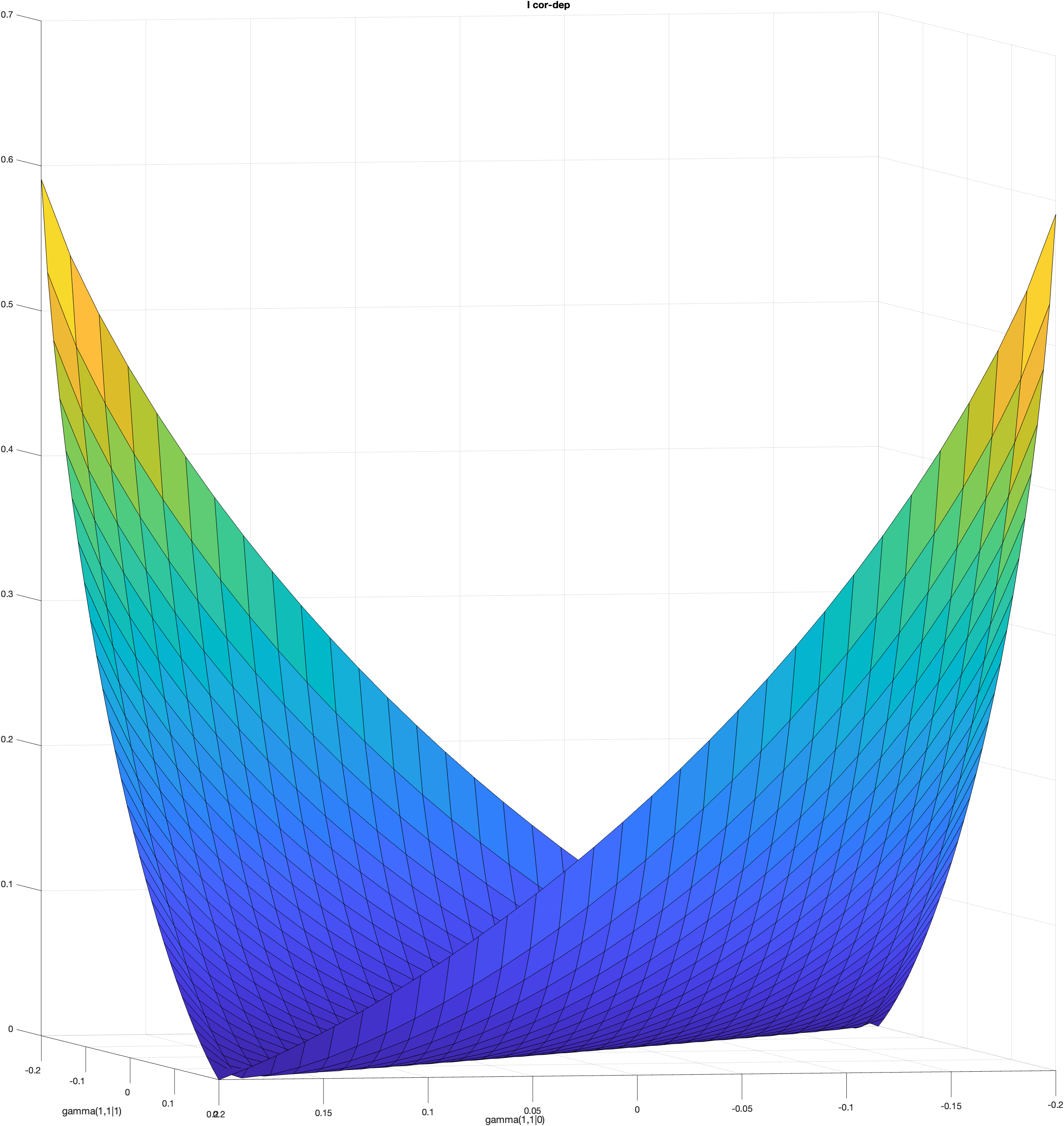}
    \caption{Landscape of mutual information when the marginals are independent.}
    \label{fig:enter-label}
\end{figure}

The analogy we have in mind is of neuroscientific nature. If an experimental neuroscientist observes neurons $X$ and $Y$ separately, they would never have information about the stimulus $S$. This effect, that there is information that would otherwise be unknown without observing both, is what one want to call the \textit{synergy}. 

Among theoretical neuroscientists, there have been (heated) debates on how to quantify synergy between two random variables (e.g. neurons) about a third one (e.g. stimulus) under the framework of Shannon's information theory, given the physical interpretation that information can be quantified with bits and has set-function-like properties. Several initial guesses (add citations) all turn out to have counter-intuitive properties. (See \cite{schneidman2003synergy} for a history review.) This problem eventually evolves into the \textit{Partial Information Decomposition (PID)} problem asking to write down the mutual information in question, $I(S:X,Y)$ as a sum (considered as the analogue of disjoint union of sets) of \textit{synergy}, \textit{redundancy} between $X$ and $Y$ (analogous to set intersection), and \textit{unique information} from $X$ and $Y$ respectively about $S$ (analogous to set symmetric difference) that satisfy expected properties to be taken as axioms. In the literature, $X$ and $Y$ are referred to as the \textit{sources} and $S$ the \textit{target}. The initially proposed axioms in \cite{williams2010nonnegative} leaves one freedom for the choice of the decomposition and thus initiated the design of several candidates. 

After years more of trial and error by several groups, a now-popular choice of PID was proposed in \cite{bertschinger2014quantifying}, commonly known as the \textit{BROJA} PID by the initials of the authors, which satisfies an operational property about the unique information and has advantage such as nonnegativity, which preserves the mental image that bits of information have set-like behaviours, and has gained been used in many practical studies, albeit still unsatisfactory such as the lack of direct generalisation beyond bivariate cases. To compute the decomposition, one needs to solve a convex optimisation problem over an affine subspace, $\Delta_P$, of the probability simplex for $(S,X,Y)$. In this paper, we see how studying this optimisation domain leads (back) to the question: how does interaction between neurons or neural populations affect the amount of information they give?

In neuroscience, interactions are often considered from two angles \cite{latham2005synergy}: the \textit{unconditional} dependence $P(X,Y)$ and the \textit{conditional} dependence $P(X,Y|S)$. The former is also known as the \textit{signal correlation}, where as the latter is called the \textit{noise correlation}. The names might not meet modern understanding anymore, but they remain widely accepted. However, most theoretical studies from the neuroscience community stay in the community and consider only small sets of restricted scenarios or local approximations, sometimes even lacking exact quantitative formulations. One main purpose of this paper is present and examine existing models. In particular, in Section \ref{section landscape}, we present observations of interplay between aforementioned correlations on Shannon's mutual information by working through basic cases based on \cite{rauh2019properties}, and in Section \ref{section GMM}, we see, in a manner of soft discussions, how these results draw parallel with Gaussian mixture models. 

In section \ref{section AG}, we explain the algebro-geometric nature of the model and how classical theorems such as the B\'ezout theorem are useful for the model. Furthermore, in Section \ref{section discriminant}, for the toy model, we write down the exact algebraic \textit{discriminants} that separate distributions of distinct behaviours, computing the exact percentage of them. To be more specific, we need to find hyperplanes that intersect \textit{some} Segre variety in a line. By doing so, we find descriptions of probability simplices more suitable for product structures (i.e. taking joint distributions) and mixture models. 



As a by-product, in section \ref{section PID}, we see the decomposition of mutual information proposed in \cite{panzeri1999correlations}\cite{pola2003exact} in fact coincides with the BROJA PID in certain principals, how they are different, and provide a ``translation'' between the two. 

We hope this paper serves as a communication between communities by introducing considerations from theoretical neuroscience and examining them from the mathematical point of views.

\section{Landscape of mutual information}\label{section landscape}
    Keep the neurons and stimulus analogy in mind. We ask how the unconditional distribution $P_{X,Y}$ and the conditional distribution $P_{X,Y|S}$ of the two neurons affect the observer (e.g. downstream neurons) to discern which stimulus $S = i$ was presented. Trusting information theory, the mutual information $I(S:X,Y)$ is the ``quantity of information'' that is conveyed through the channel. 

    As mentioned in the introduction, correlations (more concretely, under some measurement) from $P_{X,Y}$ and $P_{X,Y|S}$ are often referred to as the signal and noise correlations respectively in neuroscience. For experimentalists, it is instinctive to consider the Pearson correlation or other similar measurements for \textit{real} random variables. However the Pearson correlation does not always distinguish statistic/stochastic dependence from independence. Information theory instead deals with abstract distributions, and realisations of distributions as real random variables are extra structures not intrinsic to information theory. For example, two joint distributions on a product of state spaces of real numbers can have Pearson correlations of opposite signs but the same mutual information, because as abstract distributions, they are identical up to permutation. What parametrisation is useful is a central topic in information geometry established by Amari. (C.f. \cite{amari2000methods}.) Disregarding this principal could result in extra assumptions. 
    
   When finishing this manuscript, the paper \cite{hu2014sign} was recommended to the authour. The stories are very similar yet complementary. A reading of it is highly recommended to the reader who finds interests in this paper as well. 
    \vspace{0.5cm}

    We start with elementary examples that however contain some of the core ideas we shall encounter over and over again.
    \begin{example}\label{ex I(X:Y)}
        Let $X$ and $Y$ be two binary r.v.s. with state spaces  $\mathfrak{X} = \{x_1,x_2\}$ and $\mathfrak{Y} =  \{y_1, y_2\}$. The mutual information is a function(al) on the space of probability distributions over the state space $\mathfrak{X}  \times \mathfrak{Y}$. Instead of appointing a ``true'' distribution, we consider a family of joint probability parametrised by a parameter $t\in \mathbb{R}$, written in a matrix form that represents the ``table'' of the joint distribution: 
        \begin{equation}
            [P^t(x,y)]_{x\in \mathfrak{X},\; y \in \mathfrak{Y}}   = \begin{bmatrix}
                a & b \\
                c & d
            \end{bmatrix} + t\begin{bmatrix*}[r]
                1 & -1 \\
                -1 & 1
            \end{bmatrix*},
        \end{equation} where $a,b,c,d > 0$ and sum to $1$. \begin{equation}
            -a, -d \leq t \leq c, b
        \end{equation} so that within the bounds, $P^t$ remains nonnegative and a valid probability distribution. We denote this interval 
        \begin{equation}
            D = [-\min(a,d), \min(b,c)] 
        \end{equation} and also \begin{equation}
            D_+ = (-\min(a,d), \min(b,c))
        \end{equation} for the inteior. In the following, we identify $P^t$ with the matrix form whenever it is clear in the context to ease the notation.

        Note that the marginals $P^t(x)$ and $P^t(y)$ are constant for all $t$. Hence the mutual information $I_{P^t}(X:Y)$, considered as a function $I(P^t)$ on $D$, is convex in the linear parameter $t$. By convexity and the fact that $I$ is analytic in the interior of $D$, either there exists a unique minimum on $D$ or $I$ is constant. 

        By direct calculations, \begin{equation}
            \frac{dI}{dt} = \log \frac{(a+t)(d+t)}{(b-t)(c-t)} \text{ for } t\in D_+,
        \end{equation}  and
        \begin{equation}
            \frac{dI}{dt} = 0 \text{ iff } \det P^t = \det P^0 + t = 0.
        \end{equation} For the joint distribution $P^t$, the determinant is zero if and only if $X\perp_{P^t} Y$. Therefore, $I$ has a unique minimum in the interior if and only if the family $P^t$, contains such distribution within the constraints determined by $P^0$ (or equivalently, any $P^t$), i.e. $\det P^0 \in \overset{\circ}{D}$. Otherwise, the unique minimum occurs at the boundary of $D$. 
        
        One can easily see both situations are possible. E.g. taking $P^0$ to be rank one with nonzero entries, we have that $P^0$ is the unique interior minimum of $I$. On the other hand, if \begin{equation}
            [P^0] = \begin{bmatrix}
                \frac{3}{4} & \frac{1}{16} \\
                \frac{1}{16} & \frac{1}{8}
            \end{bmatrix},
        \end{equation} then $\det P^0 = \frac{23}{256} \notin D = [-\frac{1}{8}, \frac{1}{16}]$. 
    \end{example}

    Back to the example \ref{diag 0} in the introduction, we are interested to see how the mutual information $I(S:X,Y)$ can change according to the ``interaction'' between $X$ and $Y$ beyond the marginal distributions $P_{S,X}$ and $P_{S,Y}$. More specifically, if we fix the marginal distributions $P_{S,X}$ and $P_{S,Y}$, what can we say about $I(S:X,Y)$ as a function of the joint distribution $P_{S,X,Y}$ satisfying the marginal constraints? To do so, we need to parametrise the domain of such distributions, analogous to $D$ in the example above. Note that in the example, the set $\Delta := \{P^t|t \in D\}$ remains the same with any choice $P \in \Delta$ as $P^0$. (The space for the paramestrisation $D$ is not invariant.) 

    \begin{example}
        Let $S, X, Y$ be (multinomial) Gaussians with covariance matrix 
        \begin{equation}
            \Sigma = \begin{bmatrix}
                \Sigma_S & C_{SX}^T & C_{SY}^T \\
                C_{SX} & \Sigma_X & C_{XY}^T \\
                C_{SY} & C_{XY} & \Sigma_Y
            \end{bmatrix}
        \end{equation} with $\Sigma_{\bullet}$ denoting the covariance matrix of a r.v. and $C_{\bullet\bullet}$ the covariance of each pair. $\Sigma_{XY}$ denotes the covariance of $(X,Y)$, i.e. the bottom right of $\Sigma$.
        
        Assume $C_{SX}$ and $C_{SY}$ are fixed. We want to see how the mutual information 
        \begin{equation}
            I(S:X,Y) = \log\frac{\det\Sigma_S\det\Sigma_{XY}}{\det\Sigma}.
        \end{equation} varies with $C_{XY}$.

        Consider the $1$-dimensional case for example. To ease the notations, write 
        \begin{equation}
            \Sigma = \begin{bmatrix}
                a & d & e\\
                d & b &t\\
                e & t & c
            \end{bmatrix}
        \end{equation} so 
        \begin{equation}
            I = \log\frac{a(bc-t^2)}{abc + 2det - be^2 - at^2 -cd^2}.
        \end{equation} Hence 
        \begin{equation}
            \frac{dI}{dt} = \frac{2t}{t^2-bc} + \frac{2(at - de)}{abc + 2det - be^2 - at^2 -cd^2}.
        \end{equation} Direct computation gives that the derivative is zero if 
        \begin{equation}
            g(t) := (de)t^2 + (be^2 + cd^2)t + bcde = 0,
        \end{equation} which has two roots \begin{equation}
            t = \frac{be}{d} \text{ or } \frac{cd}{e}.
        \end{equation}

        By Sylvester's criterion, we may assume $a>0$ and $\det\Sigma_{SX} = ab - d^2 > 0$, and the constraint for $f$ is then a quadric
        \begin{equation}
            \det\Sigma = (abc - be^2 -cd^2) + 2(de)t  - at^2  > 0,
        \end{equation} giving a bounded open interval $D$. 
        
        For example, taking 
        \begin{equation}
            \Sigma = \begin{bmatrix}
                1 & 0.5 & 0.5\\
                0.5 & 1 &t\\
                0.5 &t& 1
            \end{bmatrix},
        \end{equation} one can check that the roots of $g(t)$ does not lie in $D$, and so $I$ can only attain minimum at the boundary. Meanwhile, taking $S\perp X$ and $S\perp Y$, or $d = e = 0$, $I = 0$ at $t = 0$, which lies in the interior of the domain $f \in D = (-1,1)$. The exact conditions for each case to happen is straightforward to work out, but we only give examples for the sake of illustration.    
    \end{example}

\subsection{Over the domain of correlations}
    \textbf{Set up and convention}: we fix the finite state spaces $\mathfrak{S}$, $\mathfrak{X}$, and $\mathfrak{Y}$ for r.v.s $S,X,Y$ resp. We write $I_Q$ or $I(Q)$ for a distribution $Q$ if it is clear from the context. Moreover, since we do not appoint a ``true'' distribution, $P$ only refers to the marginals $P_{S,X}$ and $P_{S,Y}$, and any joint distribution in $\Delta^{\mathfrak{S} \times \mathfrak{X} \times \mathfrak{Y}}$ will be denoted $Q$. Some examples in neuroscience are those when $S$ is the stimulus, $X$ and $Y$ two neurons, or two neural ensembles presented as neural codes or averaged firing rates.
    
    \begin{definition}
        Given $P$, we consider the \textbf{correlation domain}, \begin{equation}
        \Delta_P = \{Q \in \Delta^{\mathfrak{S} \times \mathfrak{X} \times \mathfrak{Y}}| Q_{S,X} = P_{S,X} \text{ and } Q_{S,Y} = P_{S,Y}\},
        \end{equation} which (under the usual embedding of the probability simplex in Euclidean spaces) is defined by linear constraints. In fact it is a convex domain. (C.f. \cite{bertschinger2014quantifying}\cite{rauh2019properties}) Its relative interior is denoted $\Delta_{P+}$. 
    \end{definition}

    \begin{definition}
        Given $P\in \Delta^{\mathfrak{S} \times \mathfrak{X} \times \mathfrak{Y}}$, the \textbf{shuffle distribution} is defined for all states of $(S,X,Y)$ such that \[
            Q^0(P)(s,x,y) = P(s)P(x|s)P(y|s).
        \]
    \end{definition} 

    We can immediately see some properties of $Q^0$.
    \begin{lemma} Let $P$ be a marginal distribution as the above.
    
        \begin{enumerate}
            \item $Q^0(P) \in \Delta_P \neq \emptyset$.
            \item $Q^0(Q)\in\Delta_P \forall Q \in \Delta_P$.
            \item $X\perp Y | S$ at $Q^0$.
        \end{enumerate}
    \end{lemma}
    In short, shuffle distributions are in 1-1 correspondence with correlation domains. We call $Q^0(P)$ the \textit{shuffle distribution of $P$ or of $\Delta_P$}. We say $\Delta_P$ has \textit{full support} if $Q^0$ does, i.e. nonzero on all states.
    
    There are at lease two occasions in which the shuffle distribution shows up in practice: when it is not possible to observe both (or multiple) r.v.s simultaneously, or deliberate shuffling of data as a reference point with no conditional dependence. 

    The question we focus on next is the following: where is the minimiser of information $Q^*$? Does it lie in the interior or the boundary? Since $Q(s) = P(s)$ is fixed, $I$ is convex on $\Delta_P$, so the location gives an initial description about the landscape of the function $I$. Moreover, a minimum always exists but is not necessarily unique and can occur in the (relative) interior of $\Delta_P$ or its boundary. In the following, we observe and characterise different cases of the location of the minimum. Characterisation of the cases certainly may facilitate computation time. Moreover, we shall see what these characterisation imply and their physical interpretations, which provides a more complete picture for some models widely considered by neuroscience community about neural interactions and information. 
    
    \subsubsection{Parametrising and orienting the correlation domain}

    Here we mention the results in \cite{bertschinger2014quantifying} regarding how we parametrise $\Delta_P$ by adapting and specialising \cite{hocsten2002grobner}.

    \begin{proposition}\label{parametrisation}
    Given the marginal $P$ and $\Delta_P$ as above.
        \begin{enumerate}
            \item $\Delta_P$ is the intersection of an affine space, $Q^0 + \ker \pi_P$ for the linear map $\pi_P$ that computes the (formal) marginals for a vector in $\mathbb{R}^{\mathfrak{S} \times \mathfrak{X} \times \mathfrak{Y}}$, and the probability simplex $\Dsxy$ and hence a polytope.
            
            \item $\Delta_P$ splits into a product space whose components are obtained by conditioning on the event $\{S=s\}$ for each $s$:
            \[
            \Delta_P \cong \prod_{s \in \mathfrak{S}} \Delta_{P,s}
            \] with 
            \[
            \Delta_{P,s} = \{Q \in \Delta^{\mathfrak{X}\times \mathfrak{Y}}|Q(x) = P(x|S=s), \; Q(y) = P(y|S=s)\}
            \] by the mapping 
            \[
            Q \mapsto (Q(\cdot|s_1), \cdots, Q(\cdot|s_n))
            \]

            \item The linear space $\ker \pi_P$ is spanned over $\mathbb{R}$ by all vectors of the form \[
                V_{(s,x,y);(s,x',y')} := (\delta_{(s,x,y)} + \delta_{(s,x',y')}) - (\delta_{(s,x',y)} + \delta_{(s,x,y')})
            \] for $x\neq x'$ and $y \neq y'$, where $\delta_m$ denotes the characteristic function at the state $m$.

            \item A choice of basis for $\ker \pi_P$ is \[
                \mathcal{B}_{x_0,y_0} = \{V_{(s,x_0,y_0);(s,x,y)}|x\neq x_0, \; y \neq y_0 \}.
            \] for fixed $(x_0, y_0)$.
        \end{enumerate}
    \end{proposition}

    To find out critical values of $I$, one computes the directional derivative. 
    \begin{lemma}\label{deriv I}
        The directional derivative each $V_{(s,x,y);(s,x',y')}$ is 
        \begin{equation}
        (D_{V_{(s,x,y);(s,x',y')}} I)(Q)  = \log\frac{Q(s,x,y)Q(s,x',y')}{Q(s,x,y')Q(s,x',y)}\frac{Q(x,y')Q(x',y)}{Q(x,y)Q(x',y')} 
        \end{equation} in the interior of $\Delta_P$. 
    
    
        In particular, the derivative at $Q^0$ is 
        \begin{equation}
            \log\frac{Q^0(x,y')Q^0(x',y)}{Q^0(x,y)Q^0(x',y')} 
        \end{equation} in any nondegenerate direction if $Q^0$ has full support. 
    \end{lemma} We also state some slightly more general formulae for derivatives of entropy in the appendix \ref{deriv entropy}.

    The derivatives of $I$ on the boundary do not necessarily exist but can be obtained as the limit of the above formula. 
  
    \begin{example}
        Let's consider a case with degenerate $\Delta_P$. Let $\mathfrak{S} = \{s,s'\}$ and \begin{equation}
            [Q^0(x,y|s')] = \begin{bmatrix}
                p& q  \\ 
                0 & 0
            \end{bmatrix} ,
        \end{equation} which has no room to move. So we may write \begin{equation}
            Q^t = Q^0 + tP(s)V_{s}
        \end{equation} with bounds $t \in [t^{min}, t^{Max}]$. The interval is identified with $\Delta_P$.  

        For a distribution in the relative interior $Q \in \Delta_{P+}$,
        \begin{align}
            \frac{dI(Q^t)}{dt} & = \log\frac{Q(s,x_1,y_1)Q(s,x_2,y_2)}{Q(s,x_1,y_2)Q(s,x_2,y_1)}\frac{Q(x_1,y_2)Q(x_1,y_2)}{Q(x_1,y_1)Q(x_2,y_2)} \\
            & = \log\frac{(Q(s,x_2,y_1) + q)Q(s,x_2,y_2)}{Q(s,x_2,y_1)(Q(s,x_2,y_2)+p))} \\
            & = \log\frac{1 + \frac{q}{Q(s,x_2,y_1)}}{1+\frac{p}{Q(s,x_2,y_2)}}, 
        \end{align} assume $Q^0(s,\cdot, \cdot)$ has full support, so the last equality holds. The last expression is zero if and only if \begin{equation}
           \frac{q}{Q(s,x_2,y_1)} =\frac{p}{Q(s,x_2,y_2)}
        \end{equation} That is to say, the conditional distribution $Q_{X,Y|S= s'}$ determines the condition an interior minimum can occur, but whether the domain contains such point is determined by the conditional distribution $Q_{X,Y|S}$.

        Again one can construct examples of $P$ such that either case holds. 
    \end{example}
     
    \subsubsection{Binomial correlations}\label{Amari corr}

    Looking at \ref{deriv I}, the expressions \begin{equation*}
    	 \log\frac{Q(s,x,y)Q(s,x',y')}{Q(s,x,y')Q(s,x',y)}, \; \log\frac{Q(x,y')Q(x',y)}{Q(x,y)Q(x',y')} 
    \end{equation*} measures how much probability is concentrated on the states corresponding to $+1$ or $-1$ in the parametrisation \ref{bivariate bin Delta P} and hence can be considered as a measurement for correlations. Some motivations for the definitions that follows would  probabily be more apparent in Sec. \ref{section AG}. However these are wheels reinvented. In \cite{amari2006correlation}, Amari had already proposed the same measurements as natural information theoretic correlations since in terms of information geometry, they are \textit{orthogonal} to the model and therefore can be considered independently without implicitly moving in the model. Indeed, one can check the space $\pi_P$ as in Prop. \ref{parametrisation} is orthogonal to the independence model $\Sigma$ that we shall introduce in Sec. \ref{section AG}.
    
    \begin{definition} \label{bin corr def}
        Let $S,X,Y$ be as above. Given a binomial $b$ in the commutative ring $R= \mathbb{R}\left[\mathfrak{X}\times\mathfrak{Y}\right]$ of (formal) polynomials of states of $(X,Y)$ over $\mathbb{R}$ of the form \[
            b = r_1 r_2\cdot \cdots r_k - r_{k + 1}\cdot\cdots r_{l-1}r_{l},
        \] for $r_1, \cdots, r_l \in R$ the \textbf{log binomial noise correlation} of $b$ given $s \in \mathfrak{S}$ at the distribution $Q\in\Dsxy$ is \[
            \beta_Q(b|s) = \beta(Q)(b|s) = \log Q(r_1|s)Q(r_2|s)\cdots Q(r_k|s) - \log Q(r_{k+1}|s)\cdots Q(r_l|s)
        \] provided none of the probabilities is zero. 
        
        If $b$ is a homogeneous polynomial, i.e. both terms have the same degree, then the conditioning $Q(r|s)$ can be changed to joint $Q(r,s)$. 

        The \textbf{log binomial signal correlation} of $b$ at the distribution $Q\in\Dsxy$ is \[
            \alpha_Q(b) = \alpha(Q)(b) = \log Q(r_1)Q(r_2)\cdots Q(r_k) - \log Q(r_{k+1})\cdots Q(r_l)
        \] provided none of the probabilities is zero. 

        At the shuffle distribution, $\alpha_0 := \alpha(Q^0)(b) $ is called the \textbf{shuffled signal correlation}. 
    \end{definition}
    
    The use of binomials is consistent with \cite{hocsten2002grobner}. Derivatives like in \ref{deriv entropy} give some ideas why we need such generality. In fact, we often consider ones that define critical points of mutual information so the zero loci of them coincide with the set of critical points. The direction $V_{(s,x,y);(s,x',y')}$ corresponds to the binomial $b = (x,y)(x',y') - (x',y)(x,y')$. By choosing a basis for $\Delta_P$ as in the previous section, we implicitly choose an ``orientation'' by considering the binomial correlations for $b$. Note, on boundaries of the probability simplex, if the derivative exists, it is obtained by dropping the zero terms. In this case the binomial correlation is still convenient to use as discriminant that separates distributions having different behaviours, as we shall later in the paper. 
    
    The definition for the signal correlation is somewhat less precise across the literature. Indeed we directly use the one from \cite{latham2005synergy}. Broadly, a measurement of it should reflect if the responses of neurons are similar across stimuli. For real valued r.v.s $X,Y$, we expect to see the mean of $X,Y$ given each $S = s$ correlated.  For binary $X,Y$, this is to say , for example $X,Y$ has positive signal correlation if, after identifying the states as $0,1$, the joint probability $P_{X,Y}$ is concentrated on the ``diagonal'', i.e. the states $(1,1)$ and $(0,0)$.   
    
     The following lemma corroborates this definition. 
    
    \begin{lemma}\label{n = 2 sig corr}
    	For $n = 2$, in the binary source model, $\alpha(Q^0)$ has the same sign of  \[
    	(P(y|s_1) - P(y|s_2))(P(x|s_1) - P(x|s_2)),
    	\] meaning $X$ and $Y$ have similar or dissimilar ``preference'' for $S$.
    \end{lemma}
    
    \begin{proof}
    	See appendix \ref{n = 2 sig corr proof}.
    \end{proof}
    
    We quote the following lemma to show that this is particularly consistent with the concept for the binary source model. Some objects will be defined in the next section.
    
    \begin{lemma} (\cite{amari2010conditional})
    	The secant variety of the independence model $\sigma$ in $\Delta^3$ contains the whole simplex. That is, any distribution $Q_{X,Y}$ for the binary source model can be the marginal of  some shuffle distribution $Q^0_{X,Y}$.
    	
    \end{lemma}

    \subsection{Bivariate binary source model}\label{two binary}
    In this section, we consider the case $|\mathfrak{X}| = |\mathfrak{Y}| = 2$ and $\mathfrak{S} = \{1, \cdots, n\}$. Each $\Delta_{P,s}$ is one dimensional.  
    \begin{equation}\label{bivariate bin Delta P}
        \Delta_P = \{Q^0(P) + \sum_s t_sV_s \;| \; t_s \in \left[t_s^{min}, t_s^{Max}\right]\},
    \end{equation} is an $n$-dimensional box, where $  V_s := V_{(s,x_1,y_1);(s,x_2,y_2)} $  corresponding the binomial $$b = (x_1,y_1)(x_2,y_2) - (x_1,y_2)(x_2,y_1).$$ 
    We fix the notations 
    $$\alpha(Q) := \alpha(Q)(b), \; \beta_s(Q) = \beta(Q)(b|s),$$ 
    $$ t^{min} = (t^{min}_1, \cdots, t^{min}_n),  \text{ and } t^{Max} = (t^{Max}_1, \cdots, t^{Max}_n).$$

    With example in neuroscience again, when we consider that in a short window of time, neurons either fire once or not at all. We may also consider it as an approximation for general $X, Y$ with finite states by partitioning each state spaces into two parts. (For instance, ``significant'' response and ``ignorable'' response.)

    We prove results extending those from \cite{rauh2019properties} with the assumptions from the previous section, showing how the location of the minimiser $Q^*$ is affected by properties of the marginals $P_{S,X}$ and $P_{S,Y}$. 

    \begin{lemma}\label{alpha equal beta}
        The derivative \ref{deriv I} is zero if and only if $\beta_{Q^*}(b|s) = \alpha_{Q^*}(b)$ for all $s$.  
    \end{lemma}
    \begin{proof}
        This is apparent from \ref{deriv I}.
    \end{proof}

    With the definitions of binomial correlations and the derivative formula, we can say: \textit{when noise and signal correlations are equal, information is minimised}. 

    \begin{theorem}\label{binary main theorem} Assume $\Delta_P$ has full support. 
    \begin{enumerate}
        \item   If $Q^*$ is an optimiser with coordinates $t^* := (t_1^*, \cdots, t_n^*)$ on $\Delta_P$, then the signs of all $t_i^*$ ($+, -,$ or $ 0$) are equal to that of $\alpha(Q^0)$.

        \item If $Q^*$ is on the boundary, then either $t^*_s  = t^{Max}_s$ for all $s$ or $t^{min}_s$ for all $s$. In this case, for some (choice of notations) $x\neq x'$ $y\neq y'$, $P(x|s)P(y|s) \geq P(x'|s)P(y'|s)$ for all $s$. 

        \item $Q^*$ is not unique if and only if $X\perp S$ and $Y\perp S$ on $\Delta_P$. In this case, $I(Q)$ is zero on the diagonal $\overline{t^{min}t^{Max}}$ containing $Q^0$. 
    \end{enumerate}
    \end{theorem}

    This generalises \cite[Sec. 5]{rauh2019properties} on binary $S$. 

    \begin{proof}
        \begin{enumerate}
            \item By \ref{alpha equal beta}, since $\alpha$ is independent of $s$, all $\beta(Q^*)(b|s)$ must have the same sign as $\alpha(Q^*)$, and hence the same is true for $t^*_s$. By convexity, it also determines the derivative at $Q^0$.

            \item By Lemma 3.3 in \cite{rauh2019properties}, if $Q^*$ is on the boundary, $Q^*(x,y) = 0$ for some state $(x,y)$. This means the $t_s$ coordinate for $Q^*$ is $t_s^{min}$ or $t_s^{Max}$ for all $s$. This can only happen when, for some choice of notation $x\neq x'$ and $y\neq y'$, $Q^0(x,y|s) = P(x|s)P(y|s) \geq Q^0(x',y'|s) = P(x'|s)P(y'|s)$ for all $s$. 

            \item We use results \ref{rank condition} and notations involved that we introduce later in the paper since the computation is the same. 

            If the optimum is not unique, then as in Lemma 3.1 of \cite{rauh2019properties}, the set of optimisers contains a line. By (2), the line segment can only be $L:=\overline{t^{min}t^{Max}}$. By (1),  $L$ can not intersect points whose coordinates are of different signs. 

            Now $I_{ci} = 0$, $I_{cd}$ vanishes along the diagonal $L$, implying $X\perp S$ and $Y\perp S$. 
        \end{enumerate}
  
    \end{proof}

    Note that the three cases are not mutually exclusive. Case 3. is included in Case 1 and 2. Later we shall see in Lemma \ref{no boundary} that a boundary minimiser can not have derivative zero.

 \subsection{Algebraic information geometry}\label{section AG}
    The conditions considered in \ref{two binary} are all polynomial--in fact binomial conditions. (The conditions become linear if we take logarithms. We may consider either structure depending on the situation.) Algebraic geometry is the subject that studies solutions of polynomial systems. It is used implicitly in \cite{rauh2019properties}. In this section, we exploit further this point of view to see not only how geometry can help us have a clear mental picture but how classical results in algebraic geometry can be employed to prove results otherwise requiring more computations. 
    
    For an introduction for algebraic geometry in statistical problems, see for example \cite{huh2014likelihood}. A classic introduction to algebraic geometry is \cite{fulton2008algebraic}. 
    \vspace{0.3cm}
    
    For now we consider probability simplices under the standard Euclidean embedding with coordinates in matrix forms. 

    \subsubsection{Linear conditions of fixed marginals}
        Consider two r.v.s $X$ $Y$ with $\Delta := \Delta^{\mathfrak{X}\times\mathfrak{Y}}$. The distributions $Q\in \Delta$ such that $X$ has some fixed marginal $P(X)$ satisfies a linear equation.
        \begin{equation}
            [Q(x,y)]\begin{bmatrix}
                1 \\
                1\\
                \vdots \\
                1
            \end{bmatrix} = 
            \begin{bmatrix}
                P(X = x_1) \\
                P(X = x_2) \\
                \vdots \\
            \end{bmatrix}
        \end{equation}
        Hence it is equivalent to say $Q$ lies on a hyperplane intersecting $\Delta$. If $Q(Y) = P(Y)$ is also fixed, the constraint space is obtained by intersecting further with another hyperplane. 

        For the binary models where $|\mathfrak{X}| = |\mathfrak{Y}| = 2$, this gives the line segment bounded by $\Delta$ given in \ref{ex I(X:Y)}. Since the hyperplanes for fixing the marginals of $X$ and $Y$ are parallel families respectively. There is an ($1$-dimensional) direction in which the distribution $Q(X,Y)$ can move with marginals fixed.

    \subsubsection{Segre variety of independence distributions}
        Distributions  $\{Q|X\perp_Q Y\}\subset \Delta := \Delta^{\mathfrak{X}\times\mathfrak{Y}}$ are defined by the equation $\det Q = 0$. This is called the \textit{Segre variety} $\Sigma$ (strictly speaking, its real nonnegative part). It is also the image of the mapping 
        \begin{equation}
            \begin{split}
               \sigma: \Delta^{\mathfrak{X}}\times\Delta^{\mathfrak{Y}} & \hookrightarrow  \Delta^{\mathfrak{X}\times\mathfrak{Y}} \\
                    (Q_X, Q_Y) & \mapsto  Q_XQ_Y^\top,
            \end{split}
        \end{equation} where $Q_X$ and $Q_Y$ are in form of column vectors.

        For binary $X,Y$, the mapping is 
        \begin{equation}
            (\begin{bmatrix}
                p \\ 1-p
            \end{bmatrix}, 
            \begin{bmatrix}
                q \\ 1-q
            \end{bmatrix}) \mapsto \begin{bmatrix}
                pq & p(1-q) \\
                (1-p)q & (1-p)(1-q)
            \end{bmatrix}.
        \end{equation} $\Sigma$ is a \textit{doubly ruled surface} that through each point it contains two lines corresponding to fixing marginals of $X$ and $Y$, i.e. constant $p$ or $q$, respectively. 

    \subsubsection{Mixture models as point configurations}
        We have always had mixture models in mind. 
        \begin{equation}
            \Dsxy \cong \Delta^{\mathfrak{S}} \times \prod_s \Delta^{\mathfrak{X}\times\mathfrak{Y}},
        \end{equation} via the mapping 
        \begin{equation}
            Q \mapsto (Q_S, Q_{X,Y|S= s_1}, \cdots, Q_{X,Y|S= s_n})    
        \end{equation}
        Therefore it is convenient to identify the components of conditional distributions, and a distribution $Q$ in $\Dsxy$ consists of data $Q(s)\in \Delta^{\mathfrak{S}}$ and a configuration of $n= |\mathfrak{S}|$ points of $\Delta^{\mathfrak{X}\times\mathfrak{Y}}$.  We call this a \textit{configuration of conditionals} for now. We write $q_{ij} = Q(x_i,y_j)$ for the coordinates of this simplex.  

        A configuration of conditionals determines canonically an affine (degree $1$) mapping $ r$ of $\Delta^{\mathfrak{S}}$ by mapping each vertex $\delta_s \in \Delta^{\mathfrak{S}}$ to $Q_{X,Y|S= s}\in \Delta^{\mathfrak{X}\times\mathfrak{Y}}$. The distribution $Q_{X,Y}$ is therefore identified with the image of $P_S\in \Delta^{\mathfrak{S}}$ in this mapping. In other words, giving the marginal $P$ is the same as giving a point $P_S\in \Delta^{\mathfrak{S}}$ and $n$ points $Q_1, \cdots, Q_n \in \Sigma$. 
        
        For $X,Y$ binary, $\Delta_P$ is the product of the $n$ line segments $L_j$ through each $q_j$ in the direction the fixed the marginals. Moreover, the Segre variety $\Sigma$ separates $\Delta^{\mathfrak{X}\times\mathfrak{Y}}$ into two \textit{toric chambers} corresponding to two parities of binomial correlations. The term comes from that $\Sigma$ is a toric variety, i.e. defined by binomials. We only use the term to be consistent with the literature, e.g. \cite{allman2015tensors} for now.) The necessary condition in \ref{binary main theorem} translates to the requirement that in one of the chambers the end of all line segments through $Q_1, \cdots, Q_n \in \Sigma$ hit the same wall of $\Delta^{\mathfrak{X}\times\mathfrak{Y}}$.
        
        By B\'ezout's theorem, in $\Delta^{\mathfrak{X}\times\mathfrak{Y}}$ a line, e.g. the (\textit{Zariski closure} of the) image of$ \Delta^{\mathfrak{S}}$, either intersects a degree $2$ surface, e.g. $\Sigma$, in at most two points or is contained in the surface. This gives a first application of classical algebraic geometry below.
        
       \begin{lemma}
           For $n=2$, $Q^0 = Q^*$ iff $X\perp S$ or $Y \perp S$. 
           
           Moreover, the parity of $\alpha(Q^0)$ is  for any $P_S \in \Delta^{\mathfrak{S}}_+$ with fixed $P_{X,Y|S}$. 
       \end{lemma}

       \begin{proof}
            If $\Delta_P = \{Q^0\}$, this is trivial. Otherwise $Q^0$ is in the relative interior and is a minimiser iff $\alpha(Q^0) = 0$, or $Q^0 \in \Sigma$, so $r(\Delta^{\mathfrak{S}})\subset \Sigma$. The mapping $\sigma$ shows that this holds when $X\perp S$ or $Y \perp S$.
        \end{proof}   

       \begin{remark}
           For $n\geq3$, in general $\Delta^{\mathfrak{S}}$ is mapped to a polytope of dimension $\geq 2$. For $n = 3$, the interior of the image of $\Delta^{\mathfrak{S}}$ either intersects $\Sigma$ in one degree $2$ curve or two lines (again by B\'ezout theorem). For example, consider the configuration 
           \begin{equation}
               \{\frac{1}{4}\begin{bmatrix}
                   1 & 1 \\
                   1 & 1
               \end{bmatrix},
               \frac{1}{16}\begin{bmatrix}
                   9 & 3 \\
                   3 & 1
               \end{bmatrix},
               \frac{1}{16}\begin{bmatrix}
                   3 & 1 \\
                   9 & 3
               \end{bmatrix}.\}
           \end{equation} One can check the image of $\Delta^{\mathfrak{S}}$ intersects both chambers. (E.g. it suffices by B\'ezout theorem by checking that midpoints of each sides have determinants of different signs. ) In this case, the parity of the noiseless binomial signal correlation depends on $P_S$ and in particular $Q^0$ is a minimiser only on a measure zero set.
           
           One the other hand, for the configuration  
           \begin{equation}
               \{\frac{1}{4}\begin{bmatrix}
                   1 & 1 \\
                   1 & 1
               \end{bmatrix},
               \frac{1}{16}\begin{bmatrix}
                   9 & 3 \\
                   3 & 1
               \end{bmatrix},
               \frac{1}{16}\begin{bmatrix}
                   1 & 3 \\
                   3 & 9
               \end{bmatrix}\}
           \end{equation} the image of $\Delta^{\mathfrak{S}}$ lies in the chamber of positive determinant: the interior of  $\Delta^{\mathfrak{S}}$ can only intersect in one irreducible curve of degree $2.$
       \end{remark} 
 
    \subsubsection{Level set of binomial correlation}\label{discriminant introduction}
        The condition for $\Delta_P$ to have an interior minimiser for $I$ is similar. We consider binary $X,Y$ and strictly positive distributions for simplicty. We have seen that this condition is the existence of some constant $c$ such that
       \begin{equation}
           \alpha(Q) = c = \beta_s(Q)
       \end{equation} for all $s$. This means all $Q_{X,Y|S= s}$ and $Q_{X,Y}$ lie on the same \textit{correlation level set}, i.e. the quadric surface defined by 
       \begin{equation}\label{sigma eq}
           q_{11}q_{22} - r q_{12}q_{21} = 0, 
       \end{equation} where $r = e^{c}$. Varying $c\in (-\infty, \infty)$, we have a family $\Sigma_r$ of surfaces with $\Sigma_1 = \Sigma$. Using B\'ezout theorem again, this means given $P$, this is equivalent to say the plane $H$ containing two lines $L_1, L_2$ corresponding to fixed marginals intersects with some $\Sigma_a$ in a line. We explain this in detail in the next section.

 \subsection{Discriminants on shuffle distributions}\label{section discriminant}
 
    In this section, we work out the baby model of the scheme described in Sec. \ref{discriminant introduction} in detail, finding out the exact algebraic condition for a distribution to contain a minimum in the interior of the correlation domain $\Delta_P$ for the base model where $S,X,Y$ are all binary. Note that the close form solutions in \cite{rauh2019properties} split into two cases and is only valid when $I(S:X,Y)$ has a critical point. Otherwise the minimiser lies at boundaries. It is certainly possible to derive the same results using the conditions directly--indeed we obtain four linear boundaries, conceptually equivalent to the linear condition that the solution lies within the rectangular domain. However the methods used here might provide a more systematical recipe for general cases. 
    
    In \cite{rauh2019properties}, the percentages of distributions attaining unique minima for several sizes of state spaces for $S,X,Y$ are estimated with the Monte Carlo method. (Note that for the all-binary model, as discussed, the minimum is unique except for the \textit{diagonal} $\{(P,P)\}\subset\Delta^1\times\Delta^1$.) With the discriminant we derive below in this section, it is possible to calculate analytically the percentage of distributions attaining interior or boundary minimal mutual information. Such result is perhaps of more interests as a first step to study the statistics in nature. For example, are pairs of neurons more inclined to have interior or boundary minimal information? Do they stay near or away from the minimum? If they could have an interior minimum, are they inclined to have higher or lower correlations relative to the minimum? And what are good ways to quantify these effects? 
    
    \vspace{0.5cm}
    \noindent\textbf{Assumptions and conventions}:
    \begin{enumerate}
        \item The probability simplex $\Delta^{\mathfrak{X}\times\mathfrak{Y}} \cong \Delta^n$ as the real nonnegative part of the projective space $\mathbb{P}^n_{\mathbb{C}} =: \mathbb{P}^n$ instead of the real part defined by the equality that coordinates sum to one for the convenience of homogeneous coordinates. We refer to the Segre variety $\Sigma$ ambiguously as the whole variety or the semi-algebraic set in the simplex when it is clear from the context. 

        \item We may even only consider the interior of $\Delta^n$, denoted $\Delta^n_+$, consisting of all strictly positive distributions. On the boundary, each $\Sigma_r$ consisits of two planes (or empty). In fact we shall see soon in Lemma \ref{no boundary} that we lose nothing.  
        
        \item We write the homogeneous coordinates for $\mathbb{P}^3$ in the form of $2\times 2 $ matrix \[ \begin{bmatrix}
            x & y \\
            z& w
        \end{bmatrix}. \] For the dual projective space $\mathbb{P}^{3*}$ that parametrises hyperplanes in $\mathbb{P}^3$, we use homogeneous coordinates \[\begin{bmatrix}
            a & b \\
            c & d
        \end{bmatrix}.\] 

        We use $A\cdot B := Tr AB^T$ for the usual inner product.

        \item Note that as discussed, the distribution $P_S$ does not affect whether the minimum is in the interior and is not relevant for the all-binary case. 

    \end{enumerate}

    Polynomial computations are partially carried out with the Matlab function \texttt{gbasis}. 

    \vspace{0.5cm}
    First, we fix a point $P$ on the Segre surface $\Sigma$. What are the points on $\Sigma$ such that the distribution determined has an interior minimiser? This is equivalent to ask that for $P\in\Sigma$, if for some $r \in (0, \infty)$, there is some hyperplane $H$ that (i) contains $P$ and $Q$, (ii) contains the direction (or line) \[ V := \begin{bmatrix*}[r]
        1 & -1 \\
        -1 & 1
    \end{bmatrix*}, \] and (iii) the intersection $H\cap\Sigma_r$ is linear. We may relax the first condition to (i') $H\ni P$ to have the scheme below:
    \begin{enumerate}
        \item Fix $P\in \Sigma$ and $r\in (0, \infty)$. Find the hyperplane $H_r$ satisfying (i') (ii) and (iii). (It is easy to see $r$ is unbounded on $\Delta$.)
        \item Find the intersection $C(P,r) : = H_r\cap\Sigma$.
        \item Now for fixed $P$, we have the locus of $Q$, $A(P): = \bigcup_r C(P,r)$. The points configurations with interior minimiser is thus \begin{equation}
            \mathcal{A} := \bigsqcup_{P \in \Sigma} A(P) \subset \Sigma\times\Sigma.
        \end{equation}
    \end{enumerate} Note that, by Thm. \ref{binary main theorem}, an interior minimiser is unique except for the degenerate case where the two points in the configuration on $\Sigma$ are the same point, so we do not make distinction between cases with unique minimisers and the degenerate case. Moreover, Thm. \ref{binary main theorem} also indicates a boundary minimiser is not contained in any of $\Sigma_r$. In fact:
    
    \begin{lemma}\label{no boundary}
        If $Q^*$ is a unique minimiser of $I$ at the boundary, then derivatives of $I$ at $Q^*$ are nonzero. 
    \end{lemma}

    \begin{proof}
        By \ref{binary main theorem}, the point configuration of $Q^*$, denoted $\{P,Q\}$, lies in the interior of a face (of dimension $1$ or $2$), denoted $F$, of $\Delta^3$. We can check the derivatives exist and are in the form of dropping the zeros terms from \ref{deriv I}.  

        If $\dim F = 1$, say $F = \overline{\delta_{r}\delta_{r'}}$ for some states $r, r'\in\mathfrak{X}\times\mathfrak{Y}$, the binomial correlation of the binomial $r-r'$ is monotone on $F$, so the derivatives are zero iff $P=Q$.

        If $\dim F = 2$, the binomial correlation on $F$ is defined via a degree $2$ binomial of the form $rr''-r'$ for some states $r, r' , r''$. By B\'ezout's theorem, the derivative is zero only if $P = Q$. 
    \end{proof}

    Next we want to see what hyperplanes satisfy condition (iii).  Write the equation for $H$ as $ax + b y + cz + dw = 0$ for \[ n := \begin{bmatrix}
            a & b \\
            c & d
        \end{bmatrix}.\] 

    \begin{lemma}
        Let $r>0$. Then $H\cap\Sigma_r$ is linear iff $H\in\mathbb{P}^{3*}$ lies in the locus defined by $r ad = bc$.
    \end{lemma}

    \begin{proof}
        First note $\Sigma_r$ is isomorphic to the standard Segre variety $\Sigma_1$ via a linear coordinate change so there are two lines that pass through each point on it contained in $\Sigma_r$. With the knowledge about $\Sigma$, we know each line $L$ has ideal \begin{equation}\label{line eq}
          I :=  <\mu x - \lambda y, r\mu z = \lambda w>
        \end{equation} for some $\lambda, \mu$. Since $L$ is contained in the intersection, $ax + b y + cz + dw$ must be divisible by the linear generators of $I$ above, which means $[r a:b] = [c:d]\in \mathbb{P}^1$, or $rad = bc$.
    \end{proof}
    
    The conditions are translated into algebraic equations on $\mathbb{P}^{3*}$: 
    \begin{equation}\label{H eq}
        \begin{cases}
            n\cdot V = a - b- c +d =0 \\
            rad = bc \\
            n \cdot P = 0
        \end{cases}   .         
    \end{equation} Note as discussed, there are generically two such hyperplanes.

    Next, using the \textit{corner (of cube, ad hoc name) coordinates} for $\Delta^1\times\Delta^1$, we have \begin{equation}
       \sigma : (s,t) := ([s : 1 - s], [ t : 1 -t]) \mapsto 
        \begin{bmatrix}
            s \\
            1 - s
        \end{bmatrix} \begin{bmatrix}
            t &
            1-t
        \end{bmatrix}
        = \begin{bmatrix}
            st & s(1-t) \\
            (1-s)t & (1-s)(1-t)
        \end{bmatrix} \in \mathbb{P}^3. 
    \end{equation} 

    \begin{lemma}
        For any hyperplane $H \subset \mathbb{P}^3$ containing $V$, the preimage of the intersection $H\cap\Sigma$ is a line on  $\Delta^1\times\Delta^1$ under the corner coordinates.
    \end{lemma}

    We prove a more general version in \ref{linear intersect}.


    Now fix $\Tilde{P} = (s,t)\in \Delta^1\times\Delta^1$ such that $\sigma(\Tilde{P}) = P$. For each $r\in(0,\infty)$, we have (at most) two lines as the preimage of the intersection $H_r\cap\Sigma$ passing through $\Tilde{P}$ with slopes \begin{equation}
        \delta_{\pm}(r) : = -\frac{b-d}{c-d}, 
    \end{equation} with some labelling with $\pm$, abbreviated $\delta$ when there is no confusion. 

    We may ignore the cases where $P$ lies on the boundary of $\Delta^3$. 
    \begin{lemma}\label{a neq 0}
        The system \ref{H eq} has a solution such that $abcd = 0$ iff $P$ lies at a vertex of  $\Sigma\cap\Delta^3$.
    \end{lemma}

    \begin{proof}
        If $a=0$, then the two solutions are $b = 0, c=d$ and $c = 0, b = d$, implying $s=1$ and $t=1$. By symmetry, $(\Rightarrow)$ is proved.

        Conversely, say $s = t =1$. Then plugging into \ref{H eq}, we have $a = 0$. 
    \end{proof}

    \begin{theorem}\label{discriminant thm}
        Let $P$ have corner coordinates $(s,t)$ such that $0<s\leq t\leq\frac{1}{2}$, then $A(P)$ is the union of lines through $P$ with slope \[
            \delta \in [-\infty, -\frac{1-t}{s})\cup (-\frac{t}{1-s}, \frac{t}{s}) \cup (\frac{1-t}{1-s}, \infty] 
        \]
    \end{theorem}

        The other generic cases can be obtained by symmetry. 

    \begin{proof}
        (Sketch. See appendix for details.)

        We solve the system \ref{H eq}, e.g. by computing a Gr\"obner basis. We see $\delta_\pm$ are monotone as functions in $r$ and have limits at $0$ and $\infty$. 
    \end{proof}

    \begin{remark}
        There are likely more elegant ways to do it. We present one idea in the following. 
        
        Fix $P= \sigma(s,t)$, let $\Tilde{A}(P)^*\subset \mathbb{P}^{3*}\times\mathbb{P}^1 =: M$ denote the variety determined by \ref{H eq} for all $r\in\bar{\mathbb{R}}$. The second component of the ambient space $M$, $\mathbb{P}^1 = \mathcal{P}roj\;\mathbb{R}[r_1,r_2]$ is the (extended) space of the parameter $r$.  Now consider the the extended slope function defined by $\delta([a:b:c:d], [r_1:r_2]) = [c-d:b-d]$ that is a regular map 
        \begin{equation}
            \delta: M \to \mathbb{P}^1.
        \end{equation} What is the part of the image $\delta(\Tilde{A}(P))$ corresponding to positive $r$?
    \end{remark}

    \begin{corollary}\label{volume result}
        With the assumptions above, the extremal lines intersect the boundary of $\Delta^1\times\Delta^1$ at the four vertices and points \[
            (0,\frac{t}{1-s}), \; (\frac{s}{t},1), \; (\frac{s}{1-t},0), \text{ and } (0, \frac{t-s}{1-s}). 
        \] Hence $A(P)$ has area 
        \[
            \mathcal{A}(P) : =  \frac{1}{2} + \frac{s}{2}(\frac{t}{1-t} + \frac{s}{1-s} +\frac{1-t}{t} - 1). 
        \] Thus with uniform distribution on shuffle distributions ($(\Delta^1\times\Delta^1)\times (\Delta^1\times\Delta^1)\times \Delta^1 $), there are $\frac{2}{3}$ of shuffle distributions that attain unique interior minima.
    \end{corollary}

    \begin{proof}
        By direct calculations, we have the first two statements. 

        By symmetry, we integrate $\mathcal{A}(P)$ as a function in $(s,t)$ over the region defined by $0\leq s \leq t \leq \frac{1}{2}$ and multiply by $8$.
    \end{proof}

   A typical region $A(P)$ is shown in Fig. \ref{discriminant pic}.
   \begin{figure}[h]\label{discriminant pic} 
       \centering
       \includegraphics[width=1\linewidth]{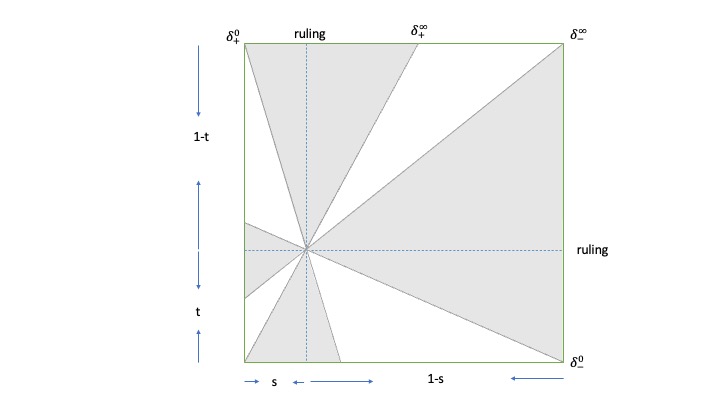}
       \caption{The region of distributions that gives an interior minimum (shaded) given a fixed distribution of coordinate $P = (s,t)$. The limiting lines are labelled with their slopes. The two rulings through $P$ are in dashed lines, separating parities of signal correlation.}
       \label{fig:enter-label}
   \end{figure}

   \begin{lemma}
       With the assumptions above, if $Q\in\Sigma$ has corner coordinates $(x,y)$, then the signal correlation of the configuration $\{P,Q\}$ has the same sign as \[
            (x-s)(y-t).
       \]
   \end{lemma}

   \begin{proof}
       Through $P$, the only two rulings are $x=s$ and $y=t$. Thus by B\'ezout's theorem, signal correlations must have the same signs within the four regions separated by these two rulings. One can easily plug in, e.g. the corners to obtain the parities. 
   \end{proof}

	The percentages of distributions attaining unique (including boundary and interior) minima for several sizes of state spaces for $S,X,Y$ are estimated with the Monte Carlo method, with respect to uniform distributions on standard simplices in \cite{rauh2019properties} instead of the ``Euclidean uniform'' measure used in \ref{volume result}. 

   For completeness, we  derive the volume form of the simplex projected onto $\Sigma$ below. 

    \begin{lemma}
        Given $P\in\Sigma$, then the length of the line segment of the line $P+gV$ with $g\in\mathbb{R}$ within the simplex $\Delta^3$ is \[
            l(s,t):=\min\{s,t,1-t,1-s\}.
        \]  Therefore if $f:\Delta^3\to \mathbb{R}$ is a real function constant on each line $P+gV$, then \[
            \int_{\Delta^3} f dP = \int_{\Delta^1\times\Delta^1} f\circ\sigma(s,t) l(s,t)dsdt,
        \] where the LHS is considered under the standard (uniform) measure, and the RHS uses the corner coordinates. 
    \end{lemma}

    \begin{proof}
        For $$P = \sigma(s,t) = \begin{bmatrix}
            st & s(1-t) \\
            (1-s)t & (1-s)(1-t)
        \end{bmatrix},$$ the length is \[
            \min\{(1-s)t, s(1-t)\} + \min\{st, (1-s)(1-t)\}.
        \] For example, taking $s<t<1-t<1-s$, this becomes $st+s(1-t) = s$. The other cases are obtained by symmetry.
    \end{proof}

    Note the standard embedded $n$-simplex has volume $\frac{1}{(n+1)!}$. One would need to normalise to obtain percentages. Theoretically we are ready to compute the percentage of distributions with interior minima. The lemmas above show that one simply needs to integrate rational functions over triangulated domains. However with the steps involved. We leave it as a future project. 

    \subsubsection{Prospects}
        We can see the cardinality of state spaces and number of observables play important rules in order for B\'ezout theorem to work. 
    
        In more general cases where $|\mathfrak{X}|, |\mathfrak{Y}|, |\mathfrak{S}| \geq 2$, each $V_{(s,x,y);(s,x',y')}$ gives a toric hypersurface (i.e. defined by a binomial). We may choose a basis among the $V$'s as the rest are redundant. The intersection of these hypersurfaces is the Segre variety parametrising all distributions $X\perp Y$. The probability simplex $\Delta^{\mathfrak{X}\times\mathfrak{Y}}$ is split into \textit{toric chambers} of parities of binomial correlations w.r.t. $V_{(s,x,y);(s,x',y')}$'s. A necessary condition for an interior distribution $Q$ to be minimising is that the configuration of points must all lie in the same chamber. The necessary and sufficient condition is that all points must lie on the intersection of some correlation level sets of the form \begin{equation}
             \cap_{x,y;x',y'} \{Q | q_{xy}q_{x'y'} - r_{x,y;x',y'}q_{x'y}q_{x'y} = 0\}.
         \end{equation} 
    
         In \cite{rauh2019properties}, cases when $S$ is binary has more detailed characterisations regarding the conditionals such as $(X,S)|Y$. It is surely interesting to see how to describe them in the geometric picture and if algebraic tools and shed more lights. 
    
         We may also consider multivariate cases for minimising $I(S:R_1, \cdots, R_l)$ with fixed marginals $P_{S,R_i}$. This corresponds to, in \cite{hocsten2002grobner}, the graph with nodes $\{S, R_1, \cdots, R_l\}$ and edges $\{(S,R_i)|i\}$. One can still see the correlation level sets are still toric hypersurfaces and in particular the intersection of zero correlation level sets is the locus of mutually independent distributions for $(R_1, \cdots, R_l)$. 

         In the appendix \ref{linear intersect}, we see for binary $S$, indeed the hyperplane intersection with the independence model is linear in the corner coordinates, making generalisation more possible. 
 
         In addition, simplices do not respect products well. The consideration of the family of Segre varieties might be a solution as a better description of joint probabilities. 
        
         We leave them as future projects.

 \subsection{Eyeball heuristics for Gaussian mixture models}\label{section GMM}


    A common topic in data science, the Gaussian mixture model (GMM) is also used often to illustrate how correlations can change information: with the same picture in mind, we have the stimulus $S = 1, \cdots, n$ and a population of $l$ neurons $R = (R_1,\cdots, R_l)$ giving Gaussian data points in $\mathbb{R}^l$ given $S = i$, i.e.  $(R|S= i) \sim \mathcal{N}(\mu_i, \Sigma_i)$. We ask how  $\mu_i$'s and $\Sigma_i$'s affect the ``discernibility'' of $S$, for example measured the mutual information $I(S:R)$. However, the mutual information has no closed form, hence often heuristics are used instead. We conduct this section only as a soft discussion, adding assumptions at will to make things comfortable, following the heuristics to see the analogy between binary models and Gaussian mixture models with more thorough exploration of parameters and also for completeness. See for example \cite{averbeck2006effects}.

    We restrict to the simple case where $n,l=2$ and write $R = (X,Y)$. To tell apart two stimuli $S=1$ and $2$, using a crude eyeball approximation as heuristics, one considers the ``overlap'' of the mass of the two Gaussians, for $S= 1$ and $2$ in the two ellipsoidal confidence regions respectively. (Sometimes even only the overlapping area of the ellipsoids are considered. This approximation can be extremely inaccurate with somewhat less mild parameter.) The more overlap there is, the harder to classify data points. We fix $\mu_i$'s and the diagonals of the covariance matrices $$\Sigma_i(b_i) = \begin{bmatrix}
        a_i & b_i \\
        b_i & c_i
    \end{bmatrix}$$ for each $i$, corresponding to the marginals $P_{R_i|S}$, and vary the covariance $b_i \in(-a_ic_i, a_ic_i) =: D_i$. With a quick look at the ellipse equations, we see $D_i$ is in 1-1 to the rotations of the axes, and $\det \Sigma_i = a_ic_i - b_i^2$ decreases with increasing $|b_i|$ so the distribution becomes more concentrated with the ratio between major and minor axes more extreme. 

    We assume for a while ``moderate parameters''. Let $P(S=1) \approx P(S= 2)$ and the diagonals $\Sigma_i(0) $ be not too far from identity and $\mu_2 - \mu_1$ have direction pointing around $45^\circ$ with magnitude around the order of $a_i, c_i$. With $b_i$'s both increasing, if the major axes of $\Sigma_1$ and $\Sigma_2$ ``align'' with $\mu_2 - \mu_1$, then the overlap should be approximately maximal, or intuitively mutual information is minimised. If $b_i$'s can increased  more, overlap decreases and information increases again. However, the restriction $D = D_1\times D_2$ might not even allow $b_i$ to increase to the point that the major axes can align. In this case, information decreases till the extreme of $b_i$'s. On the other hand, if $b_i$'s decrease to negative values, overlap decreases. This is consistent with the binary model with $\alpha(Q^0) >0$, and information is minimised when $\beta_s(Q) = \alpha(Q)$. 

    \begin{figure}[h]
        \centering
        \includegraphics[width=1\linewidth]{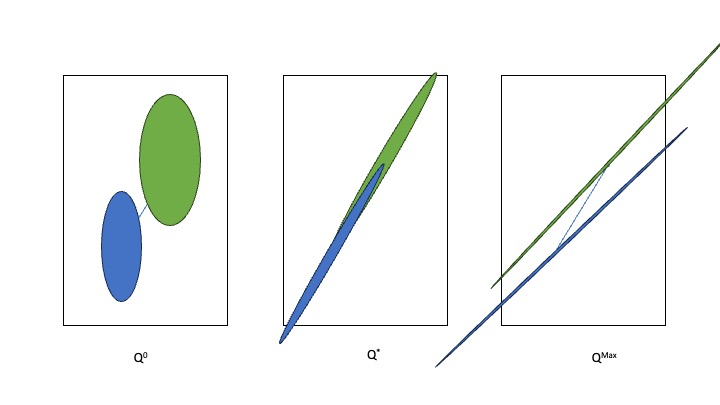}
        \caption{Illustration for overlapping Gaussians. Not fully up to scale. Re-drawn by hands for better layout as the actual elongations should be a lot more extreme. \newline
        \hspace{\linewidth}        
        \textbf{Left}: No noise correlation. This is analogous to the shuffle distribution $Q^0$. \newline
        \hspace{\linewidth}        
        \textbf{Middle}: Noise correlations ``align'' with the signal correlation, resulting in (only intuitively) maximal overlap. The is analogous to $Q^*$ in the discrete case. \newline
        \hspace{\linewidth}        
        \textbf{Right}: When noise correlations increase even more (past $Q^*$ and near $Q^{Max}$), the confidence regions separate. 
        }
        \label{fig:enter-label}
    \end{figure}

    The reader can surely experiment with more parameters to repeat the findings in the binary model. For example, when $X$ is independent of $S$, the overlap is maximal around $b_i = 0$, i.e. around the analogue of the shuffle distribution. We write down conceptual parallels between the binary bivariate models and GMMs below.
    
    \begin{table}[h]
        \centering
        \begin{tabular}{cc}
            Binary & GMM \\
             Signal correlation $\alpha$ & $\mu_2 - \mu_1 $ \\
             Noise correlation $\beta$ & $b_i$\\
             $P_{X|S}, P_{Y|S}$ & $a_i, c_i$\\
        \end{tabular}
        \caption{Parallel between GMM and binary response model in the notations used in the paper}
        \label{GMM table}
    \end{table}

    
    
    \begin{remark}
        Arguably, fixed marginal might not be the most natural setting when it comes to Gaussians. One can also consider fixing the differential entropy that is equivalent to fixing $\det \Sigma_i$. We do not pursuit this direction here. 
    \end{remark}

   
    Moreover, instead of mutual information, one can simply consider the probability of the classification error for mixture models, which might be even more straightforward for Gaussian ``overlapping'' than mutual information. 
   
    As a side note, we illustrate the minimiser of classification error can be drastically different from the minimiser of mutual information, in the following example. 
            
    \begin{example}   
        Given a mixture model consisting of random variables $S = 1,2$ and $R$ with the  distribution $P_{R,S}$ known, one may consider a straightforward classifier of data points from $R$ 
            \begin{equation}
                \hat{S}(r) = \arg\max_{s} P(S = s| R = r) = \arg\max_{s} P(S = s, R = r). 
            \end{equation} The probability of classification error is then
            \begin{equation}
                \Prob(\hat{S} \neq S) = \E [\Prob(S = 1, \Prob(S = 2|R) > \Prob(S = 1|R)) + \Prob(S = 2, \Prob(S = 1|R) > \Prob(S = 2|R))].
            \end{equation}
        
            For binary $R = 0,1$, \[
                \Prob(\hat{S} \neq S) = \sum _{s}\min(\Prob(S = s, R = 0), \Prob(S = s, R = 1) ).
             \] 

            Since we already have properties of the minimiser of mutual information on $\Delta_P$ for binary $S, X, Y$, we compare it with the minimiser of classification error when $R = (X,Y)$ on $\Delta_P \sim [t_1^{min}, t_1^{Max}]\times [t_2^{min}, t_2^{Max}]$. 

            Assume $P(S = 1) = P(S= 2) = \frac{1}{2}$ for convenience. Observe for fixed $t_2$, hence fixed $Q(x,y|S=2)$, each $Q(x,y|S=1)$ is linear (affine) as a function of $t_1$ with slope $\pm P(S=s) = \frac{1}{2}$. The classification error is thus 
            \begin{equation}
                \sum_{x,y} \min_s  Q(x,y,s).
            \end{equation} Each term in the sum is continuous and piecewise linear (or \textit{tropical} for readers familiar with tropical geometry) with slopes ordered from left to right $\frac{1}{2}, 0$ or $0, -\frac{1}{2}$ on some subintervals (could be of zero length), each occurring exactly twice. Therefore the sum consists of segments of slopes $1, \frac{1}{2}, 0, -\frac{1}{2}, -1$, and minimum must occur at one of the end points. By symmetry, this is true for fixed $t_1$, and global minima on $\Delta_P$ must be at one of the four corners. 

            We can construct examples such that the minimum occurs not at the top right or lower left corner as the case for mutual information: take $P(X|S) = P(Y|S) \cong \frac{1}{2}$, then  $(\pm\frac{1}{2}, \mp\frac{1}{2})$ are isolated minimisers but not $(\pm\frac{1}{2}, \pm\frac{1}{2})$.
    \end{example}

    For overlapping of Gaussians and classification or clustering problems, see \cite{sun2011measuring} \cite{ma2000asymptotic} and their references.

\section{Open questions and discussion}

\noindent \textbf{Amari's correlation in real data}

Quite a portion of $\Delta_P$ attain an interior minimum as we have seen. As discussed, some studies only analyse local behaviour of information as opposed to the whole landscape. Such approach can sometimes lead to inaccurate generalisations. For example, simple numerical experiments show that information can easily be much higher than at the shuffle distribution when noise correlations are higher than signal correlation. This however lacks the polynomial condition that we benefit from and does not have analytic description of the descriminant when this can happen. 

We also ask: in nature, what scenarios are more common, thinking in terms of neurons and stimulus? For example, for a tuple of neurons, is it more often the case that information is minimised within $\Delta_P$ or on the boundary? If in the interior, do neurons prefer having high positive noise correlations past the minimum, negative noise correlation, or close to the minimum? 

Amari's correlation has theoretical advantage described in the section \ref{Amari corr} and the original paper \cite{amari2006correlation}. It provides a normalised measure to compare correlations from pairs of different ``firing rates''.  For our simple binary model, we have seen that minimal information always occur on the ``diagonal''. Some test runs with real data suggested at least one issue in practice: We tested on neural spiking data of repeated simultaneous recordings, taking small time interval so that each neuron is considered Bernoulli within. Note that Amari's correlation is only continuous on nondegenrate (strictly positive) distributions. Even with sample size of $100 \sim 200$,  distributions from many pairs can still be degenerate, and further ideas are needed for reasonable application of the measure. Nevertheless, the data pairs with nondegenerate pairs do show the correlations are usually $<10$. We leave these questions as future works.

\vspace{0.5cm}
\noindent \textbf{PID without PID}

What we can possibly still agree despite historical debates about syngergy and redundancy is that they refer to higher or lower information respectively compared to some references.  

Besides the vision of building the physical model such as separation of synergy and redundancy, decompositions of information or entropy are convenient measurements in practice and have already shown usefulness for tasks in experimental neuroscience such as classification of neurons. However for this purpose, arguably they do not provide some details about the actual interactions we observe from the landscape on $\Delta_P$. Since providing the whole landscape is likely not practical, so some alternatives with simpler data structures would be necessary. 

As discussed, direct statics on the observed distribution itself could be considered. However, when the size of state spaces or when one considers beyond pairwise distributions, the computation demand might make schemes of this vein less practical.

The other alternative is to extract features of the landscape. For example, we consider Betti numbers of the level sets $I^{-1}(\tau) = \{Q| I(Q) = \tau\} \subset \Delta_P$ as a diagram in $\tau \geq 0$. This does not only record the qualitative shape of the information landscape, which as we have seen provides properties of $P$, but also the magnitudes of mutual information varying on the domain. 

There are a few major differences of this from the usual Morse theory. In the standard setting, one considers inverse image of intervals and avoid ``bad'' critical values''. However by doing so, we would not be able to see the important features of the landscape.

On the other hand, a clear model for quantified synergy or redundancy is unavailable, but for experimental neuroscience, in many cases what is needed is nothing but the relative magnitudes of information. For example, see \cite{baudot2019topological} for using higher mutual informations directly for classification tasks. In the same paper, it is also proved higher mutual information parametrises distributions up to finite points of ambiguity. 

Ultimately, does the landscape consideration provide better statistics? This needs to be testified with real applications.

\vspace{0.5cm}
\noindent \textbf{Homological nature of information theory}

We have played with some models to have a sense how ``second order interactions'' affect information. With the set-function like properties with Shannon entropy, separating parts coming from interactions of different orders is very reminiscent of homology. Indeed, in the framework of \cite{baudot2015homological}, entropy is the only cohomology class.  

\appendix

\section{Components of information}\label{section PID}
Initiated in \cite{williams2010nonnegative}, given three random variables $S,X,Y$ one wishes to find a decomposition 
\begin{equation}
	I(S: X,Y) = SI(S: X,Y) + UI(S: X \setminus Y) + UI(S: Y \setminus X) + CI(S: X,Y)
\end{equation} satisfying 
\begin{align}
	I(S,X) = SI(S: X,Y) + UI(S:X\setminus Y) & \text{ and}\\
	I(S,Y) = SI(S: X,Y) + UI(S:Y\setminus Z). &
\end{align} $CI$ is considered the synergy, or \textit{complementary information}, $SI$ the redundancy, or \textit{shared information}, whilst $UI$ represents the \textit{unique information} that one random variable has about the target $S$ excluding the other.

One can also ask for 
\begin{equation}\label{PID nonneg}
	SI, UI, CI \geq 0
\end{equation} in order to be consistent with the physical meaning that these are really ``bits of information''. Moreover, it is often taken into consideration that ($\ast$) $SI$ and $UI$ depend only on the marginals $P_{S,X}$ and $P_{S,Y}$.

These however do not uniquely determine a formulation and leaves one degree of freedom, and sometimes the assumptions are only partially taken depending on the interpretation or needs.  

\subsection{The BROJA PID}

Taking all conditions above as axioms, in  \cite{bertschinger2014quantifying}, the authors designed a PID by picking
\begin{equation}
	CI(S:X,Y) := I_P - I_{Q^*}.
\end{equation} Note although $Q^*$ is not necessarily unique, $CI$ is still well-defined. 
\begin{align}
	UI(S: X\setminus Y) = I_{Q^*}(S: X\setminus Y), \\
	UI(S: Y\setminus X) = I_{Q^*}(S: Y\setminus X), \\
	SI = I(S:X) + I(S:Y) - I_{Q^*}(S: X,Y).
\end{align}

We mention facts about this decomposition.
\begin{proposition} \label{BROJA fact}(C.f. \cite{bertschinger2014quantifying})
	\begin{enumerate}
		\item All terms $SI, UI, CI$ are nonnegative.
		\item $UI(S:X\setminus Y) = 0$ if and only if ``$X$ does not know more than $Y$'' in the sense that $\exists$ a row-stochastic matrix $\lambda\in [0,1]^{\mathfrak{X}\times\mathfrak{Y}}$ such that \[
		[P(s,x)]_{s,x} = [P(s,y)]_{s,y}\lambda. 
		\]
		\item If $X\perp_{Q^0} Y$, then $SI = 0$.
		\item The decomposition is additive for independent r.v.s.
		\item Each term is continuous on $\Dsxy$.
	\end{enumerate}
\end{proposition}

\subsection{Information break down inspired by series expansion}
Before the exact PID problem was stated, a decomposition of mutual information $I = I_{lin} + I_{ss} + I_{ci} + I_{cd}$ was first observed in \cite{panzeri1999correlations} and more finely formulated in \cite{pola2003exact}. It was first discovered by considering a Taylor expansion of $I$ under a parametrisation of the probability space essentially identical to what we introduced. The authours then grouped the series up to the second order into the following terms.    

\begin{definition}
	Given random variables $S$ and $R= (R_1, \cdots, R_k)$ with finite states denoted as $s$ and $r = (r_1, \cdots, r_k$), 
	\begin{align}
		I_{lin} & = \sum_i I(S:R_i) \\
		I_{sig-sim} & = -(I_{lin} - I_{Q^0}(S: R) ) = -D(P||P^0)\\
		I_{cor-ind} & = I(S: R) - I_{Q^0}(S: R) - I_{cor-dep} \\
		I_{cor-dep} & = D(P_{S,R}|| Q^0_{S,R}) - D(P_R||Q^0_R) \\
		& = D(P_{S|R}||Q^0_{S|R}) := \sum_{r}P(r)D(P_{S|R = r}||Q^0_{S|R = r}).
	\end{align} 
	The distribution $P_0$ in the definition of $I_{sig-sim}$ is defined as \begin{equation}
		P^0(s,r_1, \cdots) := P(s)\prod_i P(r_i).
	\end{equation}
	We also write $I_{ss}$, $I_{ci}$, and $I_{cd}$ for obvious abbreviations. 
\end{definition}
Here $D(\cdot | \cdots)$ denotes the Kullback-Leibler divergence, which we will use through out. The conditional divergence in the last line is defined consistently with the standard definition. However it is important to recall that it is not a divergence between distributions; whereas the divergence in the summand, meaning fixing $R=r$ in each term, is a divergence. In particular, $I_{cd}$ is nonnegative and $I_{ss}$ is nonpositive. 

The original formulae in \cite{panzeri1999correlations}\cite{pola2003exact} are defined with the following quantities 
\begin{align}
	\gamma_Q(r|s) & =  \begin{cases}
		\frac{Q(r|s)}{Q^0(r|s)} - 1 & \text{ if } Q^0(r|s) \neq 0 \\
		-1 & \text{  otherwise }
	\end{cases}\\  
	\mu_Q(r) & =  \begin{cases}
		\frac{Q^0(r)}{ P^0(r)} - 1 & \text{ if } P^0(r) \neq 0 \\
		-1 & \text{  otherwise }
	\end{cases}    
\end{align} named the \textit{noise} and \textit{signal correlation coefficients} respectively, where \begin{equation}
	P^0(r_1, \cdots, r_k) := \prod_i P(r_i).
\end{equation} See for example the appendix of \cite{schneidman2003synergy} for the calculation, which are intended for the bivariate case but in fact holds for general multivariate cases.

\begin{remark}
	Parametrising probabilities by realising as real random variables and use their conditional correlations often ends up taking the shuffle distribution as a base point. E.g.  \cite{mahuas2023small} also underlyingly considered essentially the same parametrisation and derived equivalent formulae to observe approximated effects of interactions between $\gamma$ and $\mu$ on the mutual information.
	
	In fact, the quantity $\gamma$ was used in \cite{panzeri1999correlations} and \cite{pola2003exact} to parametrise the probability when considering binary r.v.s. For instance, 
	\begin{equation}
		Q(r|s) = Q^0(r|s)(1 + \gamma(r|s)). 
	\end{equation} Though  $\gamma$ is defined in terms of distribution and does distinguish independence, unlike the Pearson correlations, for different states $r$, $\gamma(r|s)$'s can be dependent and often not as convenient as coordinates of the probability space. We shall see more about this. 
	
	How the Taylor series was computed exactly is however elusive to the authour. Differentiating \ref{deriv I}, the second derivative in our parametrisation should be \begin{equation}
		\frac{\partial^2 I (Q^t)}{\partial t^2 } = \sum_{s,x,y}\frac{1}{Q(s,x,y)} - \sum_{x,y} \frac{1}{Q(x,y)}.
	\end{equation}
\end{remark}     

These quantities were interpreted and named by observing interactions between the coefficients $\gamma$ and $\mu$ in the original work. We roughly translate the arguments in terms of the formulae presented above:

The \textit{linear} term, $I_{lin}$, comes from that it is the first order term in the expansion. The other three terms sum to the second order component. $I_{sig-sim}$ is considered the redundancy contributed by \textit{signal similarity}, here a synonym of signal correlation, since it is negative and measures the divergence to mutual independence. These two terms are determined solely by the marginal $P$ and thus are considered as nonsynergistic components.

The \textit{stimulus-dependent-correlation} component $I_{cor-dep}$ is zero if and only if all $\gamma(r|s)$ are independent of the stimulus $s$ (see lemmas \ref{Icd zero} and \ref{Icd vanish}), thus considered as the amount of information contributed from the dependence of noise correlation on stimuli. The leftover term $I_{cor-ind}$ is therefore considered as the contribution of correlation to synergy \textit{without} dependence on stimuli.

\subsubsection{Series expansion decomposition as a PID}
In fact, the series-expansion decomposition ``induces'' a PID in a similar spirit as the BROJA PID although without the nonnegativity axiom \ref{PID nonneg} and takes conditional independence as the baseline. We explain as follows.

First, we decide how to cut $I(S:X,Y)$ into a ``synergistic'' part, denoted $Syn$, and a ``nonsynergistic'' part, $nSyn$, the latter including redundant and unique information determined by the PID axioms. 

Choosing $nSyn = I(Q^*)$, all terms are nonnegative, and we get the BROJA PID. If, instead, with the consideration ``conditional independent means no synergy'' in mind, we choose $nSyn = I(Q^0)$ and proceed with the rest of PID axioms, we get another decomposition. In this case the redundancy is $-I_{ss}$. 

The synergy with the conditional independence is further decomposed into two terms $I_{ci} + I_{cd}$ in the series expansion decomposition. We shall see more about these two terms and how they capture some feature of the landscape of $I$.

For easy reference, we include an organised ``translation'' below.
\begin{align}\label{translation}
	CI_0 & := I_{Q^0} - I_{Q^*} \\
	I_{Q^0} & = I_{Q^*} + CI_0 \\ 
	& = I_{lin} + I_{ss} \\
	& = SI + UI_{\setminus X} + UI_{\setminus Y} + CI_0 \\
	SI + CI_0 & = - I_{ss} 
\end{align}

\begin{align}
	CI & = I_Q - I_{Q^*} = I_Q - I_{Q^0} + CI_0 \\
	& = I_{ci} + I_{cd} + CI_0, 
\end{align} 
Moreover, we shall see in \ref{Ici linear}, (the graph of) $I_{cd} + I_{Q^0}$ is tangent to $I$ at $Q^0$ on $\Delta_P$. 

\begin{figure}
	\centering
	\includegraphics[width=1.0\linewidth]{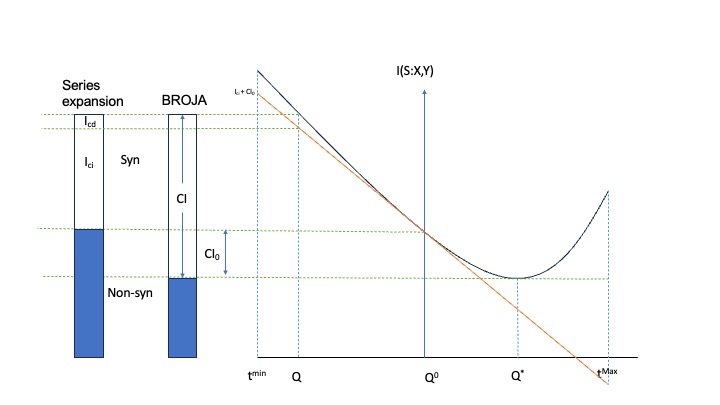}
	\caption{Illustration comparing the two decompositions. Not actual function. Hand-drawn in MS PPT. }
	\label{fig:enter-label}
\end{figure}

\subsubsection{Properties of $I_{cd}$}
$I_{cd}$ showed up with different derivations and was considered and testified several times in different occasions among several groups. We mention its vanishing properties here. They are essentially identical to lemmas in \cite{amari2006correlation}, but we state and proof in more details fit for our setting. 

\begin{lemma}\label{Icd zero} Without loss of generality, assume $P(s) > 0$ for all $s$. The following are equivalent: \\
	(1) $I_{cd} = 0$ \\
	(2) $P(s|r) = Q^0(s,r)$ if $P(r) > 0$ \\
	(3) For each $r$, $P(r) = \Tilde{\gamma}(r)Q^0(r)$ and $P(s,r) = \Tilde{\gamma}(r)Q^0(s,r)$ for some $\Tilde{\gamma}(r)$ \\
	(4) For each $r$, $P(s,r) = \Tilde{\gamma}(r)Q^0(s,r)$ for some $\Tilde{\gamma}(r)$ \\
	(5) For each $r$, $P(r|s) = \Tilde{\gamma}(r)Q^0(r|s)$ for some $\Tilde{\gamma}(r)$ \\ 
	(6) $\gamma(r|s)$ is independent of $s$.
\end{lemma} 

\begin{proof}
	The proofs are immediate for each step, noting that $\gamma(r|s) = -1$ when $P(r) = 0$. Also note that $\Tilde{\gamma} = \gamma + 1$.
	
\end{proof}

\begin{remark}
	(1) $\Leftrightarrow$ (6) is stated in \cite{pola2003exact}, where $(6) \implies (1)$ is more apparent using the original definition in terms of $\gamma(r|s)$. We include the fact in this form for completeness. 
\end{remark}

\begin{lemma}\label{Ici linear}
	On $\Delta_P$, $I_{ci}$ is a linear function in the form
	$$I_{ci} = \sum_{s, r\in \mathbf{B}}  C(r)g_{s,r}$$ for an independent set of $\{g_{s,r}|r \in \mathbf{B}\subset\mathfrak{R}\}$ that parametrises $\Delta_P$, where $C(r)$ depends on $r$ and $Q^0$ and is therefore constant on $\Delta_P$. 
	
	In particular, $I_{cd}$ is convex on $\Delta_P$. (See \ref{translation}.)
\end{lemma}

\begin{proof}
	One can either see from the original formulation in \cite{pola2003exact} or derive from the definition above to see that at $Q\in\Delta_P$ such that \begin{equation}
		Q(s,r) = Q^0(s,r) (1 + \gamma(r|s)),
	\end{equation} 
	\begin{equation}
		I_{ci} = \sum_{s,r}\gamma(r|s)Q^0(s,r)\log\frac{\prod_i P(R_i = r_i)}{Q^0(r)} = \sum_{s,r\in \mathbf{B}} g_{s,r}C(r),
	\end{equation} writing $r = (r_1, \cdots, r_k)$. The coefficient $C(r)$ is a linear sum of the logarithm terms (with coefficients $\pm 1$, see \cite{hocsten2002grobner}) and is thus constant on $\Delta_P$.
\end{proof}

Therefore we see the extra decomposition of the synergistic component separates the linear effect of noise correlations on information provided the signal correlation. 

\begin{example}
	Taking the bivariate case where $R = (X,Y)$ for a simpler example. $Q\in\Delta_P$ is parametrised as (see \ref{bivar param}) \[
	Q = Q^0 + \sum_{x \neq x_0, \: y\neq y_0}P(s)g_{(s,x,y)}V_{(s,x,y)}
	\] for fixed $x_0$ and $y_0$. The coefficients have the form 
	\begin{align}
		C(x,y)g_{(s,x,y)} & = g_{(s,x,y)}\log\frac{Q^0(x_0,y_0|s)}{Q^0(x_0,y_0)}\frac{Q^0(x,y|s)}{Q^0(x,y)}\frac{Q^0(x,y_0)}{Q^0(x,y_0|s)}\frac{Q^0(x_0,y)}{Q^0(x_0,y|s)} \\
		& = g_{(s,x,y)}\log\frac{Q^0(x,y_0)Q^0(x_0,y)}{Q^0(x_0,y_0)Q^0(x,y)}.  
	\end{align}
	Since $CI = I_{ci} + I_{cd} + CI_0$. 
	
	Alternatively, we can recover these coefficients $C(r)$ by taking the negative of the derivatives of $I(S:X,Y)$ on $\Delta_P$ at $Q^0$.  In other words, the graph of $I_{ci} + CI_0$ is tangent to the graph of $CI$ at $Q^0$. 
\end{example}

\begin{proposition}\label{rank condition}
	Consider the bivariate case $R = (X,Y)$, then if $\exists \; Q \in \Delta_P$  such that $Q\neq Q^0$ and $I_{cd}(Q)$ vanishes if and only if \[
	rank [P(x|s)]_{x,s} < |\mathfrak{X}| \text{ and } rank [P(y|s)]_{y,s} < |\mathfrak{Y}|.    \]
	
	Moreover, the zero locus of $I_{cd}$ is a linear space intersecting $\Delta_P$ of dimension \[(|\mathfrak{X}|- rank [P(x|s)])(|\mathfrak{Y}| - rank [p(y|s)]).\]    
\end{proposition}

\begin{proof} 
	By \ref{Icd zero}, $I_{cd}$ is zero iff $\gamma(r|s) =: \gamma(r)$ is independent of $s$.
	
	Consider a basis for $\Delta_P$, (c.f. \cite{bertschinger2014quantifying}\cite{hocsten2002grobner} ) \begin{equation}
		\{V_{s;(x_0,y_0); (x,y)} = \delta_{(s,x_0, y_0)} + \delta_{(s,x,y)} - \delta_{(s,x_0,y)} - \delta_{(s,x,y_0)} |s \in \mathfrak{S}, \; x_0\neq x \in \mathfrak{X}, \; y_0 \neq y \in \mathfrak{Y}\}
	\end{equation} for fixed $(x_0, y_0)$. Here each $\delta_{(s,r)}$ is the characteristic function of $(s,r)$ on $\mathfrak{S}\times \mathfrak{X} \times \mathfrak{Y}.$ Each $Q\in \Delta_P$ is of the form \begin{equation}\label{bivar param}
		Q = Q^0 + \sum_{x \neq x_0, \: y\neq y_0}P(s)g_{(s,x,y)}V_{(s,x,y)} 
	\end{equation} for coefficients $g_{(s,x,y)}$. In other words, evaluating $Q$ at each state $(s,x,y)$, we have \begin{equation}
		Q(s,x,y) = Q^0(s,x,y) + P(s)g_{(s,x,y)} \text{ for } x \neq x_0 \text{ and } y \neq y_0,
	\end{equation}
	\begin{equation}
		Q(s,x_0,y) = Q^0(s,x_0,y) - P(s)\sum_{y\neq y_0}g_{(s,x_0,y)} \text{ for } y \neq y_0,
	\end{equation}
	\begin{equation}
		Q(s,x,y_0) = Q^0(s,x,y_0) - P(s)\sum_{x \neq x_0}g_{(s,x,y_0)} \text{ for } x \neq x_0,
	\end{equation} and 
	\begin{equation}
		Q(s,x_0,y_0) = Q^0(s,x_0,y_0) + P(s)\sum_{x \neq x_0, \: y\neq y_0}g_{(s,x,y)} \text{ for } x \neq x_0 \text{ and } y \neq y_0.
	\end{equation}   
	
	Recall we have, assuming the condition that $\gamma(r|s) =: \gamma(r)$ is independent of $s$, \begin{equation}
		g_{(s,x,y)} = \gamma(x,y)Q^0(x,y|s)   \text{ for } x \neq x_0 \text{ and } y \neq y_0.
	\end{equation} Consequently, for $x \neq x_0$,
	\begin{equation} \label{indie g basis}
		\gamma(x,y_0)Q^0(x,y_0|s)= -\sum_{y \neq y_0} Q^0(x,y|s)\gamma(x,y)
	\end{equation} 
	\begin{equation}
		\implies \gamma(x,y_0)P(y_0) + \sum_{y \neq y_0} P(y|s)\gamma(x,y) = \sum_{y\in \mathfrak{Y}} P(y|s)\gamma(x,y) = 0.
	\end{equation} Similarly, for $y\neq y_0$,
	\begin{equation}
		\sum_{x\in \mathfrak{X}} P(x|s)\gamma(x,y) = 0.
	\end{equation} Also \begin{align}
		0 & =  \gamma(x_0,y_0)Q(x_0,y_0|s) - \sum_{x\neq x_0, \: y\neq y_0}Q^0(x,y|s)\gamma(x,y) \\
		& = \sum_{x} Q^0(x,y_0|s)\gamma(x,y_0) = \sum_{y} Q^0(x_0,y|s)\gamma(x_0,y)
	\end{align} 
	\begin{equation}\label{s indep gamma eq}
		\implies \sum_y P(y|s)\gamma(x_0,y) = \sum_{x} P(x|s)\gamma(x,y_0) = 0.
	\end{equation} In matrix form, we write \begin{equation}
		[P(x|s)]_{x,s}^\top[\gamma(x,y)]_{x,y} = [\gamma(x,y)]_{x,y}[P(y|s)]_{y,s} = 0.
	\end{equation} A nontrivial solution for $[\gamma]$ exists if and only if the rank condition in the statement holds, e.g. by consider singular value decomposition (SVD).
	
	It is easy to see the space \[\{Q| \gamma(r|s) \text{ independent of } s\} = \{Q|I_{cd}(Q) = 0\}\subset \Delta_P\] is linear (one can also obtain such by convexity and nonnegativity of $I_{cd}$), whose dimension is equal to the dimension of possible solutions of \ref{s indep gamma eq} by looking at \ref{indie g basis}, which, e.g. by considering SVD again, is equal to \[(|\mathfrak{X}|- rank [P(x|s)])(|\mathfrak{Y}| - rank [p(y|s)]).\]
\end{proof}

\begin{lemma}\label{Icd vanish}
	Consider the bivariate case $R = (X,Y)$ and furthermore $X$ and $Y$ are both binary. The following are equivalent.
	
	\noindent(1) $\exists Q \in \Delta_P$ where $\gamma(r|s)  =: \gamma(r) $ is independent of $s$, and $\gamma(x,y) \neq 0$ for some $(x,y)$\\
	(2) $X\perp S$ and $Y \perp S$ \\
	(3) $Q^0(s,x,y) = P(s)P(x)P(y)$ \\
	(4) $I_{Q^0}(S: X,Y) = 0 $ \\
	(5) There is a line through $Q^0$ in $\Delta_P$ such that $I_Q(S: X,Y)$ vanishes on this line. \\
	(6) $I_Q = 0$ for some $Q\neq Q^0$ \\
	(7) $CI = I_{cd}$ and $UI = SI = I_{ci} = I_{lin} = I_{ss} = 0$. 
\end{lemma}

\begin{proof}
	(1) $\implies$ (2) : the rank condition holds iff both $[P(x|s)]$ and $[P(y|s)]$ have rank $1$, i.e. $X\perp S$ and $Y\perp S$.

	\noindent (4) $\implies$ (5):  since $(X,Y) \perp_{Q^0} S$,  if $g_s$ is constant in $s$, then the condition remains.          
\end{proof}

This is the case shown in \ref{diag 0} in the introduction.

Often one tests such information measures with different magnitudes of variation of correlations under some quantification across $S= n$ in the mixture model to see if they capture these variations. For example in \cite{pola2003exact}, $\gamma$'s are used. However, from the lemma we see this poses conditions on $P$ especially when the state spaces are small. Using binomial correlations that is native in information theory can at least avoid this issue and conversely may serve as a measurement when one is interested to compare the strength of correlations across stimuli.

\begin{remark}
	In \cite{latham2005synergy}, it was argued that $I_{cd}$ (denoted $\Delta I$ in the paper) is ``an upper bound for the loss of information due to assuming the decoder''. (However, the versions stated in  \cite{merhav1994information}\cite{oizumi2010mismatched} might be easier to read.) We summarise the result for completeness.
	
	Considered in the sense of channel capacity, in \cite{merhav1994information}, the ``information rate for mismatched decoder'' is 
	\begin{equation}
		I^\ast(S:R) = \max_{P_S}\min_{P_{R|S}\in \mathcal{C}} I(S:R)    
	\end{equation} with the constraint set $\mathcal{C}$ corresponding to a certain given ``decoding metric''. The ``loss of information'' in \cite{latham2005synergy} is the difference of such from the information under the ``true'' probability, $I-I^\ast$.
	
	The authors argues that \begin{equation}
		I - I^\ast = \min_\beta D_{KL}(P_{S,R}||P^\beta_{S,R}),
	\end{equation} where \begin{equation}
		P^\beta(s,r) = \frac{P(r)P(s)Q^0(r|s)}{Z^\beta_r}
	\end{equation} is a one-parameter family, with $Z_r^\beta$ being a normalising constant, and $P^\beta_R = P_R$. In particular, taking $\beta = 1$, $I_{cd} = D_{KL}(P_{S,R}||P^1_{S,R}) \geq I= I^\ast$ gives an upper bound.  
	
	In our framework, the ``loss'' of information due to only considering the marginals is bounded by $CI = I(Q) - I(Q^*)$. 
\end{remark}

\section{Extra lemmas}

We include some basic formulae on derivatives of entropy for organised reference and also to emphasise the point of defining binomial correlations in \ref{bin corr def}. The proofs are by direct calculations. 

\subsection{Derivatives of entropy}\label{deriv entropy}
\begin{lemma}
    Given some discrete random variable $X = x_0, \cdots, x_n$ along with its probability simplex $\Delta^n$, the directional derivative in the direction $v = (v_0, \cdots,v_n)$ at $P\in \Delta^n$ such that $\sum v_i = 0$, i.e. $v$ is tangent to $\Delta^n$, is \[
        D_v H(Q)(X) = -\log \prod_{i = 0}^{n} Q(x_i)^{v_i}.
    \]
\end{lemma}

\begin{lemma}
    Following the setup in the previous lemma, if furthermore, $X = (A,B)$ is the joint of two random variables with states $(a_k,b_l)$, and rewrite $v = [w_{kl}]_{kl}$ accordingly, then \[
        D_v H(Q)(A) = -\log \prod_{k}  Q(a_k)^{\sum_l  w_{kl}},
    \] with $H(Q)(A)$ considered (lifted) as a functional on $\Delta^n$.
\end{lemma}

\begin{proof}
    The vector $w := (\sum_l  w_{kl})_{k}$ sum to zero and satisfies the previous lemma. 
\end{proof}

\begin{lemma}
    Continuing with the setup, \[
        D_v I(Q)(A:B) = \log\frac{\prod_{kl} Q(a_k,b_l)^{w_{kl}}}{\prod_k Q(a_k)^{\sum_l w_{kl}} \prod_{l} Q(b_l)^{\sum_k w_{kl}}}.
    \]
\end{lemma}

\begin{proof}
    Use $I(A:B) = H(A) + H(B) - H(A,B)$.
\end{proof}

\subsection{Linear intersection under the corner coordinates}

Given two random variables $X,Y$ with state spaces of size $M,N$ respectively, the probability simplex $\Delta^{MN}$ with coordinates $p_{ij}$, and the Segre embedding $\sigma:\Delta^M\times\Delta^N \to \Sigma \subset\Delta^{MN}$, let $\mathcal{V}$ be the linear subspace spanned by all
$$
    V_{(x,y);(s,x',y')} = (\delta_{(x,y)} + \delta_{(,x',y')}) - (\delta_{(x',y)} + \delta_{(x,y')}).
$$ Then for $Q\in \Delta^{MN}$, the subspace $(Q+\mathcal{V})\cap\Delta^{MN}$ consists of all distributions $Q$ such that $Q'_X = Q_X$ and $Q'_Y = Q_Y$. Here we consider a probability simplex as the nonnegative real part of a projective space. 

For $\Delta^M\times\Delta^N$, we use the corner coordinates: \begin{equation}
    \sigma: ((s_1, \cdots, s_{M-1}), (t_1, \cdots, t_{N-1})) \mapsto [s_it_j]_{i\leq M,j\leq M},
\end{equation} where $s_M = 1- \sum_{i<M}s_i $ and $t_M = 1 - \sum_{j<N}t_j$,  for $s_i, t_j, \sum_{i<M}s_i, \sum_{j<N}t_j \leq 1$ (i.e. the nonnegative cone under the hyperplane $\sum x_i = 1$.) 

\begin{lemma}\label{linear intersect}
    If $H$ is a hyperplane containing $\mathcal{V}$, then $\sigma(H\cap\Sigma)$ is linear under the corner coordinates. 
\end{lemma}

\begin{proof}
    Write $H: \sum_{i,j} a^{ij}p_{ij}$. So $a^{ij} - a^{ij'} - a^{i'j} + a^{i'j'}= 0$ for any $i,j,i',j'$. Plugging in the corner coordinates, we have that the intersection satisfies $\sum a^{ij}s_it_j = 0$. The coefficient of each $s_it_j$ term with $i<M, j<M$ is $a^{ij} - a^{ij'} - a^{i'j} + a^{i'j'} = 0$.
\end{proof}

\subsection{The tare map}

\begin{definition} 
	Given two binary r.v.s $X,Y$, their probability simplex is $\Delta^3$. The independence model is the Segre variety $\Sigma\subset\Delta^3$. There is a unique direction $V$ in which marginals $Q_X, Q_Y$ are fixed. The \textbf{tare map} is the projection \[\tau : \Delta^3 \to \Sigma\] in the direction $V$. 
\end{definition}

Recall $\sigma$ maps from $\Delta^1\times\Delta^1$ under the corner coordinates to $\Sigma$, all deemed as embedded in Euclidean spaces. 

\begin{lemma}
	Given $P,Q\in \Sigma$ and $\lambda, \mu$ such that $\lambda+\mu = 1$, then  \[ \tau(\lambda P + \mu Q) = \lambda \sigma^{-1}(P) + \mu \sigma^{-1}(Q).\]
\end{lemma}

\begin{proof}
	Write $R = \lambda P + \mu Q$ and $S = \sigma(\lambda \sigma^{-1}(P) + \mu \sigma^{-1}(Q))$ and let $P,Q$ have corner coordinates $(s,t), (s',t')$ respectively. Then one can check \begin{equation}
		 R - S = \lambda\mu (s-s')(t-t')\begin{bmatrix*}[r]
		 	1 & -1 \\
		 	-1 & 1
		 \end{bmatrix*}.
	\end{equation}
\end{proof}

\begin{lemma}
	Given $P_1, \cdots, P_k \in \Delta^3$ and $\lambda_1, \cdots, \lambda_k$ such that $\sum \lambda_i = 1$, \[\tau(\sum \lambda_i P_i ) = \tau(\sum \lambda_1\tau(P_i)).\]
\end{lemma}

\begin{proof}
	Since the lines $\{P_i + lV| l \in \mathbb{R}\}$ are paralell, we can assume $P_i\in \Sigma$ for all $i$. By the previous lemma, we can replace convex combination of two points with one on $\Sigma$, and hence the proof concludes with induction. 
\end{proof}

Likely with some more work one can show the lemmas are true for general $X,Y$, which would imply Lemma \ref{linear intersect}. It is unclear so far what simpler explanation there is that the corner coordinates provide such convenience. 

\section{Proofs}
\subsection{Proof of \ref{n = 2 sig corr}}\label{n = 2 sig corr proof}
For two matrices 
\begin{equation}
    A = \begin{bmatrix}
        a & b \\ 
        c & d
    \end{bmatrix} \text{ and }
     B = \begin{bmatrix}
        a' & b' \\
        c' & d'
    \end{bmatrix},
\end{equation}
\begin{equation}
    \det(A + B) = \det A + \det B + \det \begin{bmatrix}
        a & b \\
        c' & d'
    \end{bmatrix} + 
    \det\begin{bmatrix}
        a' & b' \\
        c & d
    \end{bmatrix}.
\end{equation}

For $S,X,Y$ all binary, 
\begin{equation}
    Q^0(x,y) =  Q^0(x,y,s_1) + Q^0(x,y,s_2).
\end{equation}
Since $X\perp_{Q^0} Y|S$, 
\begin{align}
    \det Q^0_{X,Y} &  = \begin{vmatrix}
       P(x'|s)P(y'|s)  & P(x'|s')P(y|s')  \\
        P(x|s)P(y'|s) & P(x|s')P(y|s')
    \end{vmatrix} + 
    \begin{vmatrix}
        P(x'|s')P(y'|s') & P(x'|s)P(y|s) \\
        P(x|s')P(y'|s') & P(x|s)P(y|s)
    \end{vmatrix} \\
      & = \det P_{Y|S} \det P_{X|S}.
\end{align}
Thus $\det Q^0_{X,Y} > 0$ if and only if the two determinants in the last line have the same sign. By column operation, \begin{equation}
    \det P_{Y|S} = \begin{vmatrix}
        P(y|s) & 1 \\
        P(y|s') & 1
    \end{vmatrix} = P(y|s) - P(y|s'),
    \end{equation} and same for $X$. 

 \subsection{Calculations for \ref{discriminant thm}}
    Recall the assumption $0<s\leq t\leq\frac{1}{2}$. By \ref{a neq 0}, we assume $a\neq 0$ and take $a = 1$. Gr\"obner basis calculations gives that 
    \begin{equation}
        d_\pm = \frac{-B\pm\sqrt{B^2 - 4st(1-s)(1-t)}}{2(1-s)(1-t)}, 
    \end{equation} where 
    \[
        B = s(1-s) + t(1-t) +r(s-t)^2.
    \] One can verify that $B^2 - 4st(1-s)(1-t)\geq 0$. Observe when $r\to 0$,
    \begin{equation}\label{d to 0}
        d_+\to -\frac{1-t}{s}, \; d_- \to -\frac{t}{1-s} 
    \end{equation}

    \begin{lemma}
       For $a>0$, $\lim_{x\to\infty}(x - \sqrt{x^2- a}) = 0$.
    \end{lemma}

    \begin{proof}
        \begin{equation}
            (x - \sqrt{x^2- a})(x + \sqrt{x^2- a}) = a.
        \end{equation} But $(x + \sqrt{x^2- a})$ diverges. 
    \end{proof} 

    The lemma above implies that \begin{equation}
        d_+ \to 0
    \end{equation} as $r\to\infty$. 
    
    \subsubsection{Case $s\neq t$}
    In the case $s\neq t$, one can compute to see that
    \begin{equation}
        \delta = \frac{t + (1-t)d}{s + (1-s)d}.
    \end{equation} Hence it is easy to see as $r\to 0$
    \begin{equation}
        \delta_+ \to -\frac{s}{1-t}=:\delta_+^0, \; \delta_- \to -\frac{t}{1-s} =: \delta_-^0,
    \end{equation} and as $r\to \infty$
    \begin{equation}
        \delta_+ \to \frac{t}{s}=:\delta_+^\infty, \; \delta_- \to \frac{1-t}{1-s} =: \delta_-^\infty. 
    \end{equation}

    Since $s<t$, \begin{equation}
        \frac{\partial\delta}{\partial d} = \frac{s-t}{(s + (1-s)d)^2} < 0
    \end{equation} for all $d$. Also \begin{equation}
        \frac{\partial d_\pm}{\partial r} = \frac{-(s-t)^2(-\sqrt{B^2 - 4st(1-s)(1-t)}\pm B)}{2(1-s)(1-t)\sqrt{B^2 - 4st(1-s)(1-t)}}
    \end{equation} so \begin{equation}
        \frac{\partial d_+}{\partial r}>0, \text{ and } \frac{\partial d_-}{\partial r}<0
    \end{equation} for all $r>0$. Consequently, $\delta_+$ decreases and $\delta_-$ increases monotonously in $r$. Hence $\delta_-$ ranges from $\delta_-^0 = -\frac{t}{1-s}$ to $\delta_-^\infty = \frac{1-t}{1-s}$. With the same analysis, we obtain the range for $\delta_+$. (Although to avoid infinite slopes, we consider the reciprocal $\frac{1}{\delta_+}$.)

    \subsubsection{Case $s=t$}
    Some cares need to be taken in some steps when considering $s=t$.

    Now \begin{equation}
        d = -\frac{s}{1-s}.
    \end{equation} With Gr\"obner basis for example, we have \begin{equation}
        \delta_\pm = -\frac{c-d}{b-d} = \frac{-1\pm\sqrt{(1-2s)^2+4rs(1-s)}}{1\pm\sqrt{(1-2s)^2+4rs(1-s)}},
    \end{equation} and \begin{equation}
        \delta_+^0 = -\frac{s}{1-s}, \; \delta_-^0 = -\frac{1-s}{s}, \; \delta_\pm^\infty = 1.
    \end{equation} Note \begin{equation}
        (1-2s)^2+4rs(1-s) = 1-4(1-r)s(1-s),
    \end{equation} so
    \begin{equation}
           \begin{cases}
        \sqrt{(1-2s)^2+4rs(1-s)} \leq 1, \; r\geq 1\\ 
        \sqrt{(1-2s)^2+4rs(1-s)} \geq 1, \; r\leq 1\\ 
    \end{cases}.     
    \end{equation}
    Since $f(x):\frac{x-1}{x+1}$ is monotone and increasing ($f'(x) = \frac{2}{(x+1)}$) for $x>0$,  \begin{equation}
        \delta_+ = \frac{\sqrt{(1-2s)^2+4rs(1-s)} - 1}{\sqrt{(1-2s)^2+4rs(1-s)} +1}
    \end{equation} increases monotonously in $r$ and ranges from $\delta_+^0 = -\frac{s}{1-s}$ to $\delta_+^\infty = 1$. The result for $\delta_-$ holds by symmetry.

\bibliographystyle{alpha}
\nocite{*}
\bibliography{pidref}

\end{document}